\documentclass[10pt,twocolumn,twoside,journal]{IEEEtran}
\usepackage{multirow}
\usepackage{lipsum}
\usepackage{graphicx}
\usepackage{amsmath}
\usepackage{color}
\usepackage{epsfig}
\usepackage{amsfonts}
\usepackage{amssymb}
\usepackage{amsthm}
\usepackage[usenames,dvipsnames]{pstricks}
\usepackage{epstopdf}
\usepackage{algorithm}
\usepackage{caption}
\usepackage{subcaption}
\usepackage[noend]{algpseudocode}

\newcommand{\argmin}{\mathop{\text{argmin}}}

\newtheorem{theorem}{Theorem}

\begin{document}

\title{{\huge On the Diversity of Uncoded OTFS Modulation in 
Doubly-Dispersive Channels }}
\author{G. D. Surabhi, Rose Mary Augustine, and A. Chockalingam \\
Department of ECE, Indian Institute of Science, Bangalore 560012 
}

\maketitle
\begin{abstract}
Orthogonal time frequency space (OTFS) is a 2-dimensional (2D) modulation 
technique designed in the delay-Doppler domain. A key premise behind OTFS 
is the transformation of a time varying multipath channel into an almost 
non-fading 2D channel in delay-Doppler domain such that all symbols in a 
transmission frame experience the same channel gain. It has been suggested 
in the recent literature that OTFS can extract full diversity in the
delay-Doppler domain, where full diversity refers to the number of 
multipath components separable in either the delay or Doppler dimension, 
but without a formal analysis. In this paper, we present a formal 
analysis of the diversity achieved by OTFS modulation along with 
supporting simulations. Specifically, we prove that the asymptotic 
diversity order of OTFS (as SNR $\rightarrow \infty$) is one. However,
in the finite SNR regime, potential for a higher order diversity is 
witnessed before the diversity one regime takes over. Also, the diversity 
one regime is found to start at lower BER values for increased frame sizes. 
We also propose a phase rotation scheme for OTFS using transcendental 
numbers and show that OTFS with this proposed scheme extracts full 
diversity in the delay-Doppler domain. 
\end{abstract}

\section{Introduction}
\label{sec1}
{\let\thefootnote\relax\footnote{This work was supported in part by the
J. C. Bose National Fellowship, Department of Science and Technology,
Government of India, and the Intel India Faculty Excellence Program.}}
Future wireless communication systems are envisioned to support diverse 
requirements that include high mobility application scenarios such as
high-speed train, vehicle-to-vehicle, and vehicle-to-infrastructure 
communications. The dynamic nature of wireless channels in such scenarios 
makes them doubly-dispersive, with multipath propagation effects causing 
time dispersion and Doppler shifts causing frequency dispersion 
\cite{jakes}. Conventional multicarrier modulation techniques such as 
orthogonal frequency division multiplexing (OFDM) mitigate the effect of 
inter-symbol interference (ISI) caused due to time dispersion. However, 
the performance of OFDM systems depends significantly on the orthogonality 
among the subcarriers. Doppler shifts can destroy the orthogonality among 
the subcarriers, resulting in inter-carrier interference (ICI) which 
degrades performance \cite{ofdm1}. 

Orthogonal time frequency space (OTFS) modulation is a recently proposed 
modulation scheme \cite{otfs2}-\cite{otfs3} that uses delay-Doppler 
domain for multiplexing the information symbols instead of time-frequency 
domain as in conventional modulation schemes. OTFS modulation uses 
transformations that spread information across the time-frequency plane. 
This spreading across the time-frequency plane converts a doubly-dispersive 
channel into an almost non-fading channel in the delay-Doppler domain. The 
relatively constant channel gain experienced by all the symbols in an OTFS 
transmission frame can greatly reduce the overhead on the channel estimation 
in a rapidly time varying channel. Another attractive feature of OTFS from 
an implementation viewpoint is that OTFS modulation can be architected over 
any multicarrier modulation (e.g., OFDM) by using additional pre-processing 
and post-processing blocks \cite{otfswhitepaper}. 

OTFS has been shown to achieve significantly better error performance 
compared to OFDM for vehicle speeds ranging from 30 km/h to 500 km/h in 
4 GHz band \cite{otfs2}. It has also been shown to perform well in mmWave 
frequency bands \cite{otfs3}. Owing to the simplicity of implementation 
and robustness to Doppler spreads, several works on OTFS have started
emerging in the recent literature \cite{otfs6}-\cite{otfs9}. OTFS 
systems using OFDM as the inner core have been considered in 
\cite{otfs6},\cite{otfs7}. In \cite{otfs4},\cite{otfs5}, the robustness 
of OTFS modulation has been demonstrated in high Doppler fading channels, 
using low complexity signal detection techniques. While \cite{otfs4} 
proposed a message passing based algorithm for OTFS signal detection, 
\cite{otfs5} proposed a Markov Chain Monte Carlo based 
algorithm for detection and a pseudo-random noise (PN) pilot based scheme
for channel estimation in the delay-Doppler domain. The above detection 
algorithms were devised using a system model based on the vectorized 
input-output relation for OTFS \cite{otfs4}. Signal detection and channel 
estimation for multiple-input multiple-output OTFS (MIMO-OTFS) have been 
considered in \cite{otfs9}, where it has been shown that MIMO-OTFS offers 
significantly better performance compared to MIMO-OFDM. 
It has been highlighted in \cite{otfs2} that the OTFS 
modulation can be viewed as a generalization of TDMA and OFDM. Likewise, 
OTFS can also be interpreted as a generalization of other multicarrier
modulation schemes such as filtered multitone \cite{mtone1} and generalized 
frequency division multiplexing (GFDM) \cite{gfdm1}. Recently, it has been
shown in \cite{otfs_gfdm} that OTFS can be implemented using a GFDM 
framework. 

While recent papers on OTFS have demonstrated the performance superiority 
of OTFS over OFDM, a formal analysis and claim on the diversity order 
achieved by OTFS is yet to appear. It has been suggested in \cite{otfs1a} 
that OTFS can achieve full diversity in the delay-Doppler domain, where
full diversity refers to the number of clustered reflectors in the channel 
(in other words, the number of multipath components separable in either 
the delay or Doppler dimension). However, this suggestion has not been
established through analysis or simulation. Filling this gap, our 
contribution in this paper provides a formal analysis of the diversity 
order achieved by OTFS in doubly-dispersive channels with supporting 
simulation results. The key findings and contributions in this work 
can be summarized as follows.
\begin{itemize}
\item 	We first derive the diversity order of OTFS in a single-input 
	single-output (SISO) setting with maximum likelihood (ML) 
	detection. It is shown that that the asymptotic diversity 
	order of OTFS (as SNR $\rightarrow \infty$) is one. Though the
	asymptotic diversity order is one, potential for a higher order 
	diversity is witnessed in the finite SNR regime before the 
	diversity one regime takes over. Also, it is observed that the 
	diversity one regime starts at lower BER values for increased 
	frame sizes. A lower bound on the BER computed by summing up 
	the pairwise error probabilities corresponding to all pairs of 
	data matrices whose difference matrices have rank one provides
	an analytical support for this observation.
\item   Next, in an attempt to extract full diversity in the asymptotic 
	regime, we propose a phase rotation scheme for OTFS using 
	transcendental numbers. It is shown that OTFS with this proposed
	scheme extracts the full diversity in the delay-Doppler domain.
\item	Finally, we extend the diversity analysis to MIMO-OTFS and 
	show that the asymptotic diversity order is equal 
	to the number of receive antennas. We also extend the phase 
	rotation scheme to MIMO-OTFS system to extract full diversity 
	in the delay-Doppler domain.
\end{itemize}

The rest of the paper is organized as follows. The OTFS modulation scheme
is presented in Sec. \ref{sec2}. The diversity analysis of OTFS in SISO 
setting and corresponding simulation results are presented in Sec.
\ref{sec3}. The proposed phase rotation scheme that achieves full
diversity is presented in Sec. \ref{sec4}. The MIMO-OTFS system, its 
asymptotic diversity order, and phase rotation scheme are presented 
in Sec. \ref{sec5}. Conclusions are presented in Sec. \ref{sec6}. 

\section{OTFS modulation}
\label{sec2}
In this section, we describe OTFS modulation designed in the delay-Doppler 
domain. We first introduce the delay-Doppler representation
and the associated transforms and then present the mathematical
description of the OTFS modulation.
\subsection{Delay-Doppler representation and OTFS modulation}
\label{sec2a}
Fundamentally, a signal can be represented either as a function of time, 
or as a function of frequency, or as a quasi-periodic function of delay 
and Doppler \cite{otfswhitepaper}. These three representations are 
interchangeable by means of the canonical transforms, as depicted in 
Fig. \ref{trans_triangle}. The nodes of the triangle represent the three 
ways of representing a signal and the edges represent the canonical 
transformation used for the conversion between them. The conversion 
between the time and frequency representations is  through the Fourier 
transform, and the conversion of the delay-Doppler representation to 
the time and frequency representations is through the Zak transforms 
$Z_t$ and $Z_f$, respectively. It is important to note that the 
composition of any pair of transforms is equal to the remaining one. 
For example, the Fourier transform is a composition of two Zak transforms 
($\mbox{FT}=Z_t\circ Z_f^{-1}$). The signals in the delay-Doppler domain 
can be viewed as functions $\phi(\tau,\nu)$ on a two-dimensional 
delay-Doppler plane whose points are parametrized by $\tau$ and $\nu$. 
This representation is a quasi-periodic representation and has an 
associated delay period $\tau_r$ and a Doppler period $\nu_r$, such 
that $\tau_r \nu_r=1$. The delay-Doppler representation can be converted 
to time and frequency representations by Zak transforms  $Z_t$ and $Z_f$, 
respectively, given by 
\cite{otfswhitepaper}
\begin{equation}
Z_t(\phi)= \hspace{-1mm} \int_0^{\nu_r} \hspace{-3mm} e^{j2\pi t \nu} \phi(t,\nu) \mathrm{d} \nu, \ \
Z_f(\phi)=\hspace{-1mm} \int_0^{\tau_r} \hspace{-3mm} e^{-j2\pi \tau f} \phi(\tau,f) \mathrm{d} \tau.
\end{equation}
\begin{figure}
\centering
\includegraphics[width=5cm, height=2.5 cm]{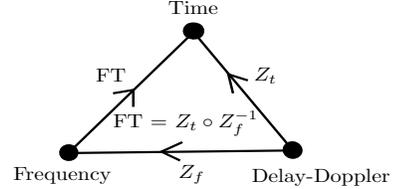}
\caption{The fundamental transform triangle.}
\label{trans_triangle}
\vspace{-4mm}
\end{figure}
A fundamental feature of OTFS modulation that distinguishes it from 
other time-frequency (TF) modulation schemes is the use of delay-Doppler 
domain for multiplexing the modulation symbols. These symbols in the 
delay-Doppler domain can be converted into time domain using the Zak 
transform $Z_t$. The transformation that uses a single Zak transform 
to convert a signal in delay-Doppler domain to a signal in time domain 
can also be carried out in two steps. That is, the signal in the 
delay-Doppler domain is first transformed to time-frequency domain, 
and the resulting time-frequency signal is converted to a time domain 
signal using a second transformation. As a consequence of this two step 
transformation, 
OTFS modulation can be implemented using simple pre- and post-processing 
steps over any multicarrier modulation scheme such as OFDM. The series 
of transformations involved in OTFS modulation transforms a time varying 
multipath channel into a slowly varying channel in the delay-Doppler 
domain. The complex baseband channel response in the delay-Doppler 
domain is denoted by $h(\tau,\nu)$, where $\tau$ and $\nu$ are the 
delay and Doppler variables, respectively. With this representation, 
the received signal $y(t)$ due to a transmit signal $x(t)$ is given by
\begin{equation}
\label{channel}
y(t)=\int_{\nu} \int_{\tau} h(\tau,\nu)x(t-\tau)e^{j2\pi\nu(t-\tau)} \mathrm{d} \tau \mathrm{d} \nu.
\end{equation} 
The channel coefficients in this representation correspond to the group 
of reflectors  associated with a particular delay depending on reflectors' 
relative distance and Doppler value depending on its relative velocity. 
Since the velocity and the relative distance remain roughly the same for 
a relatively longer duration, the delay-Doppler channel coefficients 
are time invariant for a larger observation time as compared to that in 
time-frequency representation \cite{otfs1}. Also, the delay-Doppler 
representation of the channel impulse response yields a sparse 
representation of the channel, thus requiring only fewer channel 
parameters to be estimated. With this, we now proceed to the description 
of the OTFS modulation scheme architected using pre- and post-processing 
operations over a multicarrier modulation.
\begin{figure*}[t]
\centering
\includegraphics[width=12.5 cm, height=3.0 cm]{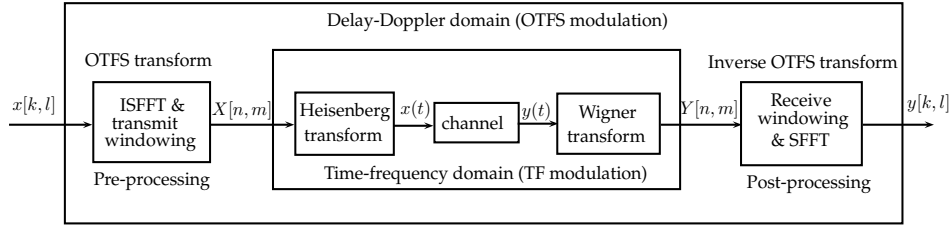}
\vspace{1mm}
\caption{OTFS modulation scheme.}
\label{blkdiag}
\vspace{-2mm}
\end{figure*}

The block diagram of the OTFS modulation scheme is shown in Fig. \ref{blkdiag}. 
The inner box in the block diagram is the familiar multicarrier TF modulation 
and the outer box with pre- and post-processor is the OTFS modulator that 
operates in the delay-Doppler domain. At the transmitter, the information 
symbols (e.g., QAM symbols) denoted by $x[k,l]$ residing in delay-Doppler 
domain are mapped to the TF signal $X[n,m]$ through the 2D inverse symplectic 
finite Fourier transform (ISFFT) and windowing. Subsequently, this TF signal 
is transformed into a time domain signal $x(t)$ through Heisenberg transform
for transmission. At the receiver, the received signal $y(t)$ is transformed 
back to a TF domain signal $Y[n,m]$ through Wigner transform (inverse 
Heisenberg transform). The TF signal $Y[n,m]$ thus obtained is mapped to 
the delay-Doppler domain signal $y[k,l]$ using the symplectic finite Fourier 
transform (SFFT) for demodulation. In the subsequent subsections, we
describe the TF modulation and the OTFS modulation in detail.

\vspace{-4mm}
\subsection{Time-frequency modulation}
\label{sec2b}
\begin{itemize}
\item The time-frequency plane is sampled at intervals $T$ and $\Delta f$, 
respectively, to obtain a 2D lattice or grid $\Lambda$, which can be 
defined as
$\Lambda=\lbrace(nT,m\Delta f), n=0,\cdots,N-1, m=0,\cdots, M-1\rbrace$.
\item The signal in TF domain $X[n,m]$, $n=0,\cdots,N-1$, $m=0,\cdots,M-1$  
is transmitted in a given packet burst, which has duration $NT$ and occupies 
a bandwidth of $M\Delta f$.
\item Let $g_{tx}(t)$ and $g_{rx}(t)$ denote the transmit and receive 
pulses, respectively. We assume $g_{tx}(t)$, $g_{rx}(t)$ to be ideal pulses 
satisfying the bi-orthogonality property with respect to translations by 
integer multiples of time $T$ and frequency $\Delta f$, i.e.,
\begin{equation}
\int e^{-j2\pi m\Delta f(t-nT)}g^*_{rx}(t-nT)g_{tx}(t) dt=\delta(m)\delta(n).
\label{biorth}
\end{equation}
The bi-orthogonality property of the pulse shapes ensures that the 
cross-symbol interference is eliminated in the symbol reception.
Although ideal pulses cannot be realized in practice, 
given the constraint imposed by the uncertainty principle, they can be 
approximated by the pulses whose support is highly concentrated in 
time and frequency \cite{otfs4}. Design of pulses concentrated in time 
and frequency to minimize the cross-symbol interference has been 
discussed in \cite{pulse_design1},\cite{pulse_design2}.
  
\item TF modulation: The signal in 
TF domain $X[n,m]$ is transformed to the time domain signal $x(t)$ through 
the Heisenberg transform, given by
\begin{equation}
\hspace{-0mm}
x(t)= \sum_{n=0}^{N-1} \sum_{m=0}^{M-1} X[n,m]g_{tx}(t-nT)e^{j2\pi m \Delta f (t-nT)}.
\label{tfmod}
\end{equation}

\item TF demodulation: At the receiver, the sufficient statistic for symbol 
detection is obtained by matched filtering the received signal with the 
receive pulse $g_{rx}(t)$. This requires the computation of the 
cross-ambiguity function $A_{g_{rx},y}(\tau,\nu)$, given by 
\begin{equation}
\label{crossambig}
A_{g_{rx},y}(\tau,\nu)=\int g_{rx} ^*(t-\tau) y(t) e^{-j2 \pi \nu(t-\tau)} \mathrm{d}t.
\end{equation}
Sampling this function on the lattice $\Lambda$ yields the matched filter 
output, given by
\begin{equation}
\label{wigner}
Y[n,m] = A_{g_{rx},y}(\tau,\nu)|_{\tau =nT,\nu =m \Delta f}.
\end{equation} 
Equation (\ref{wigner}) is called the Wigner transform, which can be 
looked at as the inverse of the Heisenberg transform.
A detailed discussion on Heisenberg and Wigner 
representations and its applications in communication 
theory has been presented in \cite{Heisenberg},\cite{rep_theory}.
\item If $h(\tau,\nu)$ has finite support bounded by
$(\tau_{\mbox{\scriptsize{max}}},\nu_{\mbox{\scriptsize{max}}})$ and if 
$A_{g_{rx}g_{tx}}(\tau,\nu)=0$ for $\tau \in (nT-\tau_{\mbox{\scriptsize{max}}},nT+\tau_{\mbox{\scriptsize{max}}})$, 
$\nu \in (m\Delta f-\nu_{\mbox{\scriptsize{max}}},m\Delta f+\nu_{\mbox{\scriptsize{max}}})$, 
$\forall n,m$ except for $n=0, m=0$ where $A_{g_{rx}g_{tx}}(\tau,\nu)=1$,
the relation between $Y[n,m]$ and $X[n,m]$ for TF modulation can be derived 
as \cite{otfs1a}
\begin{equation}
\label{tfinpop}
Y[n,m] = H[n,m]X[n,m] + V[n,m],
\end{equation}
where $V[n,m]$ is the 
noise at the output of the matched filter and $H[n,m]$ is given by
\begin{equation}
H[n,m]=\int_{\tau} \int_{\nu} h(\tau,\nu) e^{j2\pi \nu nT} e^{-j2\pi (\nu + m \Delta f) \tau} \mathrm{d} \nu \mathrm{d} \tau.
\end{equation}
\end{itemize}
From (\ref{tfinpop}), note that $X[n,m]$ is not affected by cross-symbol 
interference either in time or in frequency. In the absence of noise, the 
received symbol $X[n,m]$ is same as the transmitted symbol except for 
the complex scale factor $H[n,m]$. Note that the complex scale factor 
$H[n,m]$ is a weighted superposition of Fourier exponential functions.
This relation can be formally expressed via a two dimensional transform 
called the symplectic Fourier transform.

\vspace{-2mm}
\subsection{OTFS modulation}
\label{sec2c}
\begin{itemize}
\item 	Let $X_p[n,m]$ denote the periodized version of $X[n,m]$ with 
period $(N,M)$. The SFFT of $X_p[n,m]$ is defined as
\begin{equation}
x_p[k,l] = \sum_{n=0}^{N-1} \sum_{m=0}^{M-1} X_p[n,m] e^{-j2\pi( {nk \over N} - {ml \over M} )},
\end{equation}
and ISFFT of {\small $x_p[k,l]=SFFT^{-1} (x_p[k,l])$} is defined as 
\begin{equation}
X_p[n,m] = {1 \over MN }\sum_{k=0}^{N-1} \sum_{l=0}^{M-1} x_p[k,l] e^{j2\pi( {nk \over N}-{ml \over M})}.
\end{equation}

\item OTFS transform: The information symbols in the delay-Doppler domain 
$x[k,l]$ are mapped to TF domain symbols $X[n,m]$ as follows:
\begin{equation}
X[n,m] = W_{tx}[n,m]SFFT^{-1}(x_p[k,l]),
\label{otfsmod}
\end{equation}
where $W_{tx}[n,m]$ is the transmit windowing square summable function. 
\item$X[n,m]$ thus obtained is in TF domain and is TF modulated as 
described in the previous subsection for transmission through the channel. 
\item OTFS demodulation: 
The received signal $y(t)$ is transformed into $Y[n,m]$ using Wigner filter 
as in (\ref{wigner}). A receive window $W_{rx}[n,m]$ is applied 
to $Y[n,m]$ to obtain $Y_W[n,m]=W_{rx}[n,m]Y[n,m]$. This is periodized to 
obtain $Y_p[n,m]$ with period  $(N,M)$, given by 
\begin{equation}
Y_p[n,m] = \sum_{k,l=- \infty}^{\infty}  Y_W[n-kN,m-lM].
\label{otfsdemod1}
\end{equation}
\item The symplectic finite Fourier transform is then applied to $Y_p[n,m]$ 
to convert it from TF domain back to delay-Doppler domain to obtain  
$y_p[k,l]$ as
\begin{equation}
y_p[k,l]=SFFT (Y_p[n,m]).
\label{otfsdemod2}
\end{equation}
The output sequence of demodulated symbols is obtained
as $y[k,l]=y_p[k,l]$ for $k=0,1,\cdots,N-1$ and $l=0,1,\cdots,M-1$.\\
{\em Note:} Periodization here can be understood by taking 
the analogy of the discrete time Fourier transform (DTFT). The DTFT of a 
discrete time signal is a continuous and periodic function of frequency.
Sampling the spectrum in the frequency domain periodizes the signal in the 
time domain. Analogously, the discrete symplectic Fourier transform (DSFT) 
of a sequence is continuous and periodic \cite {DSFT}. Sampling the DSFT 
of a sequence in the delay-Doppler domain periodizes the signal in the 
time-frequency domain.
\end{itemize}
Using the equations (\ref{tfmod})-(\ref{otfsdemod2}), the input-output 
relation in OTFS can be derived as \cite{otfs2}
\begin{equation}
y[k,l]={1 \over MN} \sum_{l'=0}^{M-1} \sum_{k'=0}^{N-1} x[k',l'] h_w \left( {k-k' \over NT}, {l-l' \over M \Delta f} \right) + v[k,l],
\label{otfsinpoutp}
\end{equation}
where
\begin{equation}
h_w \left({k-k' \over NT}, {l-l' \over M \Delta f} \right) = h_w (\nu,\tau)|_{\nu={k-k' \over NT},\tau={l-l'\over M \Delta f}},
\label{deldoppchannel}
\end{equation}
where $h_w(\nu,\tau)$ is the circular convolution of the channel response 
with a windowing function $w(\nu,\tau)$, given by
\begin{equation}
h_w(\nu,\tau)=\int_{\nu'} \int_{\tau'} h(\tau',\nu')w(\nu-\nu',\tau-
\tau') e^{-j2\pi \tau\nu} \mathrm{d} \tau' \mathrm{d} \nu',
\label{circ_conv}
\end{equation}
and $w(\nu,\tau)$ is the discrete symplectic Fourier transform (DSFT) of the 
time-frequency window, defined as
\begin{equation}
w(\nu,\tau) = \sum_{m=0}^{M-1} \sum_{n=0}^{N-1} W_{tx}[n,m] W_{rx}[n,m] e^{-j2 \pi (\nu nT - \tau m \Delta f)}. 
\label{winfunc}
\end{equation}
Note that the circular convolution in (\ref{circ_conv}) does not involve 
any cyclic prefixing. Circular convolution here arises naturally from the 
OTFS pre- and post-processing (ISFFT and SFFT) operations (from Theorem 2 
of \cite{otfs1a}).

\subsection{Vectorized formulation of the input-output relation}
\label{sec2d}
Consider a channel with $P$ paths, resulting from $P$ clusters of 
reflectors, where each reflector is associated with a delay and a 
Doppler, which can be represented in delay-Doppler domain as
\begin{equation}
h(\tau,\nu) =\sum_{i=1}^{P} h_i \delta(\tau -\tau_i) \delta(\nu-\nu_i),
\label{sparsechannel}
\end{equation}
where $h_i$, $\tau_i$, $\nu_i$ represent the channel gain, delay, and 
Doppler shift associated with $i$th cluster, respectively. 
We define $\tau_i\triangleq\frac{\alpha_i+a_i}{M\Delta f}$ and 
$\nu_i\triangleq\frac{\beta_i+b_i}{NT}$, where $\alpha_i$, $\beta_i$
are integers and $a_i, b_i$ are real, where $-\frac{1}{2}<a_i,b_i \leq 
\frac{1}{2}$. We refer to $a_i$ and $b_i$ as the fractional parts of  
delay $\tau_i$ and Doppler $\nu_i$, respectively. The fractional parts 
$a_i$ and $b_i$ can be neglected if $M$ and $N$ are large, and hence the 
delay resolution $1/M\Delta f$ and the Doppler resolution $1/NT$ are 
sufficiently small to approximate the path delays and Doppler shifts to 
the nearest integer sampling points \cite{DDchannelest}.
We initially assume the fractional parts $a_i$s and
$b_i$s to be zero and carry out the diversity analysis of OTFS in 
Sec. \ref{sec3}. We extend the diversity analysis for non-zero 
fractional delays and Dopplers (i.e., $a_i,b_i\neq0$) in Appendix 
\ref{app_a}. 
Assuming $\tau_i= \frac{\alpha_i}{M\Delta f}$ and 
$\nu_i=\frac{\beta_i}{NT}$ and assuming the transmit 
and receive window function $W_{tx}[n,m]$ and $W_{rx}[n,m]$ to be 
rectangular, the input-output relation for the channel in 
(\ref{sparsechannel}) can be derived as \cite{otfs4}
\begin{equation}
y[k,l] = \sum_{i=1}^{P} h_i' x[(k-\beta_i)_N,(l-\alpha_i)_M] + v[k,l]. 
\label{inpopnofracdopp}
\end{equation}
The $h_i'$s are given by
\begin{equation}
h_i'=h_i e^{-j2 \pi \nu_i \tau_i},
\label{chan_coeff} 
\end{equation}
where $h_i$s are assumed to be i.i.d and distributed as 
$\mathcal{CN}(0,1/P)$ (assuming uniform scattering profile). The 
input-output relation in (\ref{inpopnofracdopp}) can be vectorized as
\cite{otfs4}
\begin{equation}
\mathbf{y} = \mathbf{Hx} + \mathbf{v}, 
\label{vecform}
\end{equation}
where $\mathbf{x}, \mathbf{y}, \mathbf{v} \in \mathbb{C} ^{MN\times 1}$, 
$\mathbf{H} \in \mathbb{C}^{MN\times MN}$, the $(k + Nl)$th element of 
$\mathbf{x}$, $x_{k+Nl}=x[k,l]$, $k=0,\cdots,N-1,\ l=0,\cdots,M-1$, 
and $x[k,l]\in \mathbb{A}$, where  $\mathbb{A}$ is the modulation
alphabet (e.g., QAM\hspace{-1.1mm} /\hspace{-1mm} PSK). Likewise, 
$y_{k+Nl}=y[k,l]$ and 
$v_{k+Nl}=v[k,l]$, $k=0,\cdots,N-1,\ l=0,\cdots,M-1$. 

\section{Diversity Analysis of OTFS}
\label{sec3}
Consider the vectorized formulation of input-output relation in the
SISO OTFS scheme given by (\ref{vecform}). Note that there are only $P$ 
non-zero elements in each row and column of the equivalent channel matrix 
($\mathbf{H}$) due to modulo operations. Hence the vectorized input-output 
relation in (\ref{vecform}) can be rewritten in an alternate form as
\begin{equation}
\mathbf{y}^T=\mathbf{h'}\mathbf{X}+\mathbf{v}^T,
\label{hXform}
\end{equation}
where $\mathbf{y}^T$ is $1\times MN$ received vector, $\mathbf{h'}$ is a 
$1\times P$ vector whose $i$th entry is given by 
$h'_i=h_ie^{-j2\pi\nu_i\tau_i}$, $\mathbf{v}^T$ is the $1\times MN$ noise 
vector, and $\mathbf{X}$ is a $P\times MN$ matrix whose 
$i$th column ($i=k+Nl$,\ $i=0,1,\cdots,MN-1$), denoted by $\mathbf{X}[i]$, 
is given by
\begin{align}
\mathbf{X}[i] & =\begin{bmatrix}
x_{(k-\beta_1)_N+N(l-\alpha_1)_M} \\
x_{(k-\beta_2)_N+N(l-\alpha_2)_M} \\
\vdots \\
x_{(k-\beta_P)_N+N(l-\alpha_P)_M}
\end{bmatrix}
\label{X_mat}.
\end{align}
The representation of $\mathbf{X}$ in the form given in (\ref{hXform}) 
allows us to view $\mathbf{X}$ as a $P\times MN$ symbol matrix. For 
convenience, we normalize the elements of $\mathbf{X}$ so that the 
average energy per symbol time is one. The signal-to-noise ratio (SNR), 
denoted by $\gamma$, is therefore given by $\gamma=1/N_0$. Assuming 
perfect channel state information and ML detection at the receiver, the 
probability of transmitting the symbol matrix $\mathbf{X}_i$ and deciding 
in favor of $\mathbf{X}_j$ at the receiver is the pairwise error probability 
(PEP) between $\mathbf{X}_i$ and $\mathbf{X}_j$, given by \cite{DTse}
\begin{equation}
P(\mathbf{X}_i\rightarrow \mathbf{X}_j|\mathbf{h'},\mathbf{X}_i)=Q \left( \sqrt{\frac{\|\mathbf{h'}(\mathbf{X}_i-\mathbf{X}_j)\|^2}{2N_0}} \right).
\label{PEP1} 
\end{equation}
The PEP averaged over the channel statistics can be written as
\begin{equation}
\small
P(\mathbf{X}_i\rightarrow \mathbf{X}_j)=\mathbb{E}
 \left[ Q \left( \sqrt{\frac{\gamma\ \|\mathbf{h'}
(\mathbf{X}_i-\mathbf{X}_j)\|^2}{2}} \right) \right].
\label{PEP2}
\end{equation}
This can be simplified by writing 
$\|\mathbf{h'}(\mathbf{X}_i-\mathbf{X}_j)\|^2$ as
\begin{equation}
\|\mathbf{h'}(\mathbf{X}_i-\mathbf{X}_j)\|^2
=\mathbf{h'}(\mathbf{X}_i-\mathbf{X}_j)
(\mathbf{X}_i-\mathbf{X}_j)^H\mathbf{h'}^H.
\label{norm}
\end{equation}
The matrix $(\mathbf{X}_i-\mathbf{X}_j) (\mathbf{X}_i-\mathbf{X}_j)^H$ is 
Hermitian matrix that is diagonalizable by unitary transformation. Hence 
it can be written as
\begin{equation}
(\mathbf{X}_i-\mathbf{X}_j)(\mathbf{X}_i-
\mathbf{X}_j)^H=\mathbf{U\Lambda U}^H,
\label{svd}
\end{equation}
where $\mathbf{U}$ is unitary and 
$\mathbf{\Lambda}=\mbox{diag}\lbrace \lambda_1^2, \cdots \lambda_P^2 \rbrace$, 
$\lambda_i$ being $i$th singular value of the  difference matrix 
$\mathbf{\Delta}_{ij}$, given by 
$\mathbf{\Delta}_{ij} = (\mathbf{X}_i-\mathbf{X}_j)$. Substituting 
(\ref{svd}) in (\ref{norm}), and defining 
$\tilde{\bf h}^H=\mathbf{U}^H\mathbf{h'}^H$, (\ref{norm}) can be 
simplified as
\begin{equation}
\|\mathbf{h'}(\mathbf{X}_i-\mathbf{X}_j)\|^2
=\tilde{\bf h}\mathbf{\Lambda} \tilde{\bf h}^H 
= \sum \limits_{l=1}^{r}\lambda_l^2|\tilde{h}_l|^2,
\label{norm_simplified}
\end{equation}
where $r$ denotes the rank of the difference matrix $\mathbf{\Delta}_{ij}$. 
Substituting (\ref{norm_simplified}) in (\ref{PEP2}), the average PEP 
between symbol matrices $\mathbf{X}_i$ and $\mathbf{X}_j$ can be written as
\begin{equation}
P(\mathbf{X}_i\rightarrow \mathbf{X}_j) = \mathbb{E} \left[Q\left(\sqrt{\frac{\gamma \ \sum_{l=1}^r \lambda_l^2|\tilde{h}_l|^2}{2}}\right)\right].
\label{PEP3}
\end{equation}
Note that, since $\tilde{\mathbf{h}}$ is obtained by multiplying a 
unitary matrix to $\mathbf{h}'$, it has the same distribution as that 
of $\mathbf{h'}$. Therefore, $\tilde{h}_l$s are distributed as 
$\mathcal{CN}(0,1/P)$. Using this, the average PEP in (\ref{PEP3}) 
can be simplified to get the following upper bound on PEP \cite{DTse}
\begin{equation}
P(\mathbf{X}_i\rightarrow \mathbf{X}_j)  \leq \prod 
\limits_{l=1}^{r}\frac{1}{1+\ \dfrac{\gamma \lambda_l^2}{4P}}. 
\label{PEP4}
\end{equation}
At high SNRs, (\ref{PEP4}) can be further simplified as 
\begin{equation}
P(\mathbf{X}_i\rightarrow \mathbf{X}_j)  \leq \frac{1}{\gamma^r\prod 
\limits_{l=1}^{r} \dfrac{\lambda_l^2}{4P}}. 
\label{PEP5}
\end{equation}
From (\ref{PEP5}), it can be seen that the exponent of the SNR term 
$\gamma$ is $r$, which is equal to the rank of the difference matrix
$\mathbf{\Delta}_{ij}$. For all $i,j$, $i\neq j$, the PEP with the 
minimum value of $r$ dominates the overall BER. Therefore, the achieved 
diversity order, denoted by $\rho_{\tiny \mbox{siso-otfs}}$, is given by 
\begin{equation}
\rho_{\tiny \mbox{siso-otfs}} = \min_{i,j \ i\neq j}\ \mbox{rank}(\mathbf{\Delta}_{ij}).
\label{div_order}
\end{equation}
Now, consider a case when $x_i[k,l]=a$ and $x_j[k,l]=a'$, 
$\forall k=0,\cdots,N-1$ and $l=0,\cdots,M-1$. This  corresponds to 
the case when $\mathbf{X}_i=a.\mathbf{1}_{P\times MN}$ and $
\mathbf{X}_j=a'.\mathbf{1}_{P\times MN}$. Then, 
$\mathbf{\Delta}_{ij}=(\mathbf{X}_i-\mathbf{X}_j)=(a-a').\mathbf{1}_{P\times MN}$, 
whose rank is one, which is the minimum rank of $\mathbf{\Delta}_{ij}$,
$\forall i,j$, $i\neq j$. Hence, the asymptotic diversity order of OTFS 
with ML detection is one. 
$\quad \quad \quad \quad \quad \quad \quad \ \
\quad \quad \quad \quad \quad \quad \quad \ \
\quad \quad \quad \quad \quad \quad \quad \ \
\square$ 

From the above diversity analysis, it is evident that OTFS does not extract 
full diversity in the asymptotic regime and the asymptotic diversity order 
is equal to one{\footnote{We note that the above result 
on the asymptotic diversity order of OTFS holds even for the more general 
input-output relation which considers non-zero fractional delay and Doppler 
values. This result for the case of non-zero fractional delays and Dopplers 
is derived in Appendix \ref{app_a}.}}. 
However, using a lower bound on the average BER and 
simulation results, we show next that, under certain conditions, OTFS can 
achieve close to full diversity in the finite SNR regime. 

\vspace{-2mm}
\subsection{Lower bound on the average BER}
\label{sec3a}
In this subsection, we derive a lower bound on the BER of OTFS. This
lower bound, along with simulation results in the next subsection, 
provides insight into finite SNR diversity of OTFS. For the ease of 
exposition, we assume BPSK symbols. We obtain a lower bound on BER 
by summing the PEPs corresponding to all the  pairs $\mathbf{X}_i$ and 
$\mathbf{X}_j$, such that the difference matrix
$\mathbf{\Delta}_{ij}=(\mathbf{X}_i-\mathbf{X}_j)$ has rank equal to one.
With this, a lower bound on BER is given by
\begin{equation}
\mbox{BER}  \geq \frac{1}{2^{MN}}\sum 
\limits_{k=1}^{\kappa}P(\mathbf{X}_i\rightarrow
\mathbf{X}_{j}), 
\label{lb1}
\end{equation}
where $\kappa$ denotes the number of difference matrices 
($\mathbf{\Delta}_{ij}$s) having rank one. When $\mathbf{\Delta}_{ij}$ 
has rank one, it has only one non-zero singular value ($\lambda_1$), 
which can be  computed to be $\sqrt{4PMN}$. With this, the PEP  in 
(\ref{PEP3}), for the pair $(\mathbf{X}_i, \mathbf{X}_j$) with 
$\mathbf{\mathbf{\Delta}}_{ij}$ having rank one simplifies to
\begin{equation}
P(\mathbf{X}_i\rightarrow \mathbf{X}_{j})
=\mathbb{E} \left[ Q \left( \sqrt{2\gamma PMN
|\tilde{h}_1|^2} \right) \right].
\label{PEP6}
\end{equation}   
Since $\tilde{h}_1 \sim \mathcal{CN}(0,1/P)$, evaluating the expectation in 
(\ref{PEP6}) gives \cite{DTse}  
\begin{equation}
P(\mathbf{X}_i\rightarrow \mathbf{X}_{j}) = \frac{1}{2} \left( 1-\sqrt{\frac{MN}{MN+\gamma^{-1}}}\right).
\label{PEP7}
\end{equation}
Using (\ref{PEP7}) in (\ref{lb1}), we get the lower bound as
\begin{equation}
\mbox{BER}  \geq  \frac{\kappa}{2^{MN}} \ \frac{1}{2} \left( 1-\sqrt{\frac{MN}{MN+\gamma^{-1}}}\right).
\label{lb2}
\end{equation}
At high SNRs, this can be further simplified  as
\begin{equation}
\mbox{BER}  \geq    \frac{\kappa}{2^{MN}} \frac{1}{4\gamma MN}.
\label{lb3}
\end{equation}
Observe that  (\ref{lb3}) serves as diversity one lower bound on the 
average BER and its value depends on the ratio $\dfrac{\kappa}{2^{MN}}$. 
As the values $M$ and $N$ increase, the $2^{MN}$ term dominates the ratio 
$\dfrac{\kappa}{2^{MN}}$, and therefore increasing $M$ and $N$ can reduce the 
value of the lower bound in (\ref{lb3}). We will observe this behavior in
the simulation results presented in the next subsection. Further, we  will 
also see that the BER meets the lower bound at high SNR values. This means 
that the BER can decrease with a higher slope for higher values of $M$ 
and $N$ before it changes the slope and meets the diversity one lower 
bound of (\ref{lb3}).
\begin{table}
\begin{center}
\begin{tabular}{|l|l|} \hline
\textbf{Parameter} & \textbf{Value} \\ \hline
Carrier frequency (GHz) & 4 \\ \hline
Subcarrier spacing (kHz) &  3.75 \\ \hline
Number of paths ($P$) & 4 \\ \hline
Delay-Doppler profile ($\tau_i,\nu_i$)  & $(0,0), (0,\frac{1}{NT})$, \\ 
& 
$(\frac{1}{M\Delta f},0), (\frac{1}{M\Delta f}, \frac{1}{NT})$\\ \hline
Modulation scheme & BPSK \\ \hline
\end{tabular}
\caption{Simulation parameters}
\label{SimPar}
\vspace{-6mm}
\end{center}
\end{table}

\begin{table*}
{\scriptsize 
\begin{center}
\begin{tabular}{ |c|c|c|c| }
\hline
 $\mathbf{X}_i$ & $\mathbf{X}_{j}$
& $\mathbf{\Delta}_{ij}=(\mathbf{X}_i-\mathbf{X}_{j})$\\
\hline
 $\begin{bmatrix}-1 &-1&-1&-1\\ -1&-1&-1&-1\\ -1&-1&-1&-1\\ -1&-1&-1&-1  \end{bmatrix}$ & $\begin{bmatrix}1 &1&1&1\\ 1&1&1&1\\ 1&1&1&1\\ 1&1&1&1  \end{bmatrix}$ &$\begin{bmatrix}-2 &-2&-2&-2\\ -2&-2&-2&-2\\ -2&-2&-2&-2\\ -2&-2&-2&-2  \end{bmatrix}$ \\
 \hline
$\begin{bmatrix}1 &1&-1&-1\\ 1&1&-1&-1\\ -1&-1&1&1\\ -1&-1&1&1  \end{bmatrix}$ & $\begin{bmatrix} -1 &-1&1&1\\ -1&-1&1&1\\ 1&1&-1&-1\\ 1&1&-1&-1 \end{bmatrix}$ &$\begin{bmatrix}2 &2&-2&-2\\ 2&2&-2&-2\\ -2&-2&2&2\\ -2&-2&2&2  \end{bmatrix}$ \\
\hline
 
$\begin{bmatrix}-1 &1&1&-1\\ 1&-1&-1&1\\ 1&-1&-1&1\\ -1&1&1&-1  \end{bmatrix}$
& $\begin{bmatrix}1 &-1&-1&1\\ -1&1&1&-1\\ -1&1&1&-1\\ 1&-1&-1&1  \end{bmatrix}$ & $\begin{bmatrix}-2 &2&2&-2\\ 2&-2&-2&2\\ 2&-2&-2&2\\ -2&2&2&-2  \end{bmatrix}$\\
\hline

$\begin{bmatrix}1 &-1&1&-1\\ -1&1&-1&1\\ 1&-1&1&-1\\ -1&1&-1&1  \end{bmatrix}$
& $\begin{bmatrix}-1 &1&-1&1\\ 1&-1&1&-1\\ -1&1&-1&1\\ 1&-1&1&-1  \end{bmatrix}$ & $\begin{bmatrix}2&-2&2&-2\\ -2&2&-2&2\\ 2&-2&2&-2\\ -2&2&-2&2 \end{bmatrix}$\\
\hline

$\begin{bmatrix}1 &-1&-1&1\\ -1&1&1&-1\\ -1&1&1&-1\\ 1&-1&-1&1  \end{bmatrix}$
& $\begin{bmatrix}-1 &1&1&-1\\ 1&-1&-1&1\\ 1&-1&-1&1\\ -1&1&1&-1  \end{bmatrix}$ & $\begin{bmatrix}2 &-2&-2&2\\ -2&2&2&-2\\ -2&2&2&-2\\ 2&-2&-2&2   \end{bmatrix}$\\
\hline

$\begin{bmatrix}-1 &1&-1&1\\ 1&-1&1&-1\\ -1&1&-1&1\\ 1&-1&1&-1  \end{bmatrix}$
& $\begin{bmatrix}1 &-1&1&-1\\ -1&1&-1&1\\ 1&-1&1&-1\\ -1&1&-1&1  \end{bmatrix}$ & $\begin{bmatrix}-2 &2&-2&2\\ 2&-2&2&-2\\ -2&2&-2&2\\ 2&-2&2&-2  \end{bmatrix}$\\
\hline

$\begin{bmatrix}-1 &-1&1&1\\ -1&-1&1&1\\ 1&1&-1&-1\\ 1&1&-1&-1  \end{bmatrix}$
& $\begin{bmatrix}1 &1&-1&-1\\ 1&1&-1&-1\\ -1&-1&1&1\\ -1&-1&1&1  \end{bmatrix}$ & $\begin{bmatrix}-2 &-2&2&2\\ -2&-2&2&2\\ 2&2&-2&-2\\ 2&2&-2&-2  \end{bmatrix}$\\
\hline

$\begin{bmatrix}1 &1&1&1\\ 1&1&1&1\\ 1&1&1&1\\ 1&1&1&1  \end{bmatrix}$ & $\begin{bmatrix}-1 &-1&-1&-1\\ -1&-1&-1&-1\\ -1&-1&-1&-1\\ -1&-1&-1&-1  \end{bmatrix}$
& $\begin{bmatrix}2 &2&2&2\\ 2&2&2&2\\ 2&2&2&2\\ 2&2&2&2  \end{bmatrix}$\\
\hline
\end{tabular}
\vspace{2mm}
\caption{The pair of matrices $(\mathbf{X}_i,\mathbf{X}_j)$ and the 
corresponding $\mathbf{\Delta}_{ij}$ with rank equal to one in OTFS 
with $M=N=2$ with the delay-Doppler profile given in Table \ref{SimPar}.}
\label{error_vec}
\end{center}
\vspace{-8mm}
}
\end{table*}

\subsection{Simulation results}
\label{sec3b}
In this subsection, we present the BER performance of 
OTFS modulation with ML detection. Consider the vectorized input-output 
equation in (\ref{vecform}). At the receiver, the detection is carried 
out jointly over $MN$ channel uses, using the ML detection rule given by
\begin{equation}
\hat{\mathbf x}=\argmin_{\mathbf{x} \in \mathbb{A}^{MN}} \Vert\mathbf{y}-\mathbf{H}\mathbf{x} \Vert^2.
\label{ML}
\end{equation}
 
\noindent{\em Note on the choice of $M$ and $N$ in OTFS systems:}
The delay-Doppler plane where the modulation symbols reside is 
discretized to an information grid which can be denoted by $\Gamma$, 
given by
\begin{equation}
\Gamma=\lbrace (\tfrac{k}{NT},\tfrac{l}{M\Delta f}), k=0,1, \cdots,
N-1, l=0,1, \cdots, M-1\rbrace.
\end{equation}
Here, $1/NT$ and $1/M\Delta f$ represent the quantization steps of the 
Doppler shift and the delay, respectively. For a communication system
with a total bandwidth of $B=M\Delta f$, and a latency constraint of
$T_l=NT=N/\Delta f$, the maximum supportable Doppler is $(N-1)/NT$ and 
the maximum supportable delay is $(M-1)/M\Delta f$. The parameters 
$M$ and $N$ are chosen such that the system can support the maximum
delay $\tau_{\mbox{\scriptsize{max}}}$ and maximum Doppler 
$\nu_{\mbox{\scriptsize{max}}}$, among all the channel paths, i.e.,
$\Delta f < 1/\tau_{\mbox{\scriptsize{max}}}$ and $\Delta f >
\nu_{\mbox{\scriptsize{max}}}$. The following example provides an
illustration of the choice of $M$ and $N$. \\
{\em Example 1:} Suppose the maximum delay spread and Doppler spread
of the channel are $\tau_{\mbox{\scriptsize{max}}}=1 \ \mu$s and 
$\nu_{\mbox{\scriptsize{max}}}=1$ kHz, respectively. Also, let the
system bandwidth and latency constraint be $B=10$ MHz and $T_l=1$ ms,
respectively. Then, $\Delta f$ must be such that
$\nu_{\mbox{\scriptsize{max}}}<\Delta f < 1/\tau_{\mbox{\scriptsize{max}}}$,
i.e., 1 kHz $< \Delta f <$ 1 MHz. In this range, let us take $\Delta f$ to 
be 20 kHz. Then, $B=M\Delta f$ gives 
$M=\frac{B}{\Delta f}=\frac{10\times 10^6}{20\times 10^3}=50$. Likewise, 
$T_l=\frac{N}{\Delta f}$ gives 
$N=T_l\Delta f=1\times 10^{-3}\times 20\times 10^3=20$. So the choice of
$M$ and $N$ in this system is $(M,N)=(50,20)$.

Figure \ref{UPLB} shows the simulated BER performance of OTFS with $M=2$, 
$N=2$, and BPSK. The channel model is according to (\ref{sparsechannel}), 
and the number of taps is considered to be four (i.e., $P=4$). A carrier 
frequency of 4 GHz and a subcarrier spacing of 3.75 kHz are considered. 
Other parameters considered for the simulation are given in Table 
\ref{SimPar}. The path delays ($\tau_i$s) are chosen
such that $\tau_i=\frac{\alpha_i}{M\Delta f}$ and 
$\alpha_i \in \lbrace0,\cdots, M-1\rbrace$. Similarly, Doppler shifts 
($\nu_i$s) are chosen such that $\nu_i=\frac{\beta_i}{NT}$ and 
$\beta_i \in \lbrace0,\cdots N-1\rbrace$. The maximum Doppler shift 
considered is $1/NT =1.875$ kHz for $N=2$. This corresponds to a maximum 
speed of 506.25 km/h. In addition to the simulated BER plot, we have also 
plotted the lower bound of (\ref{lb3}) and the union bound based upper 
bound for the considered system. It can be seen from the figure that the 
simulated BER, the lower bound, and the upper bound almost coincide at 
high SNR values, which means that the bounds are tight in the high SNR 
regime. Further, it can be seen that, the simulated BER shows a higher 
diversity order in the low to medium SNR regime, before it changes the 
slope and meets the diversity one lower bound.

\begin{figure}
\centering
\includegraphics[width=9.5cm, height=6.5cm]{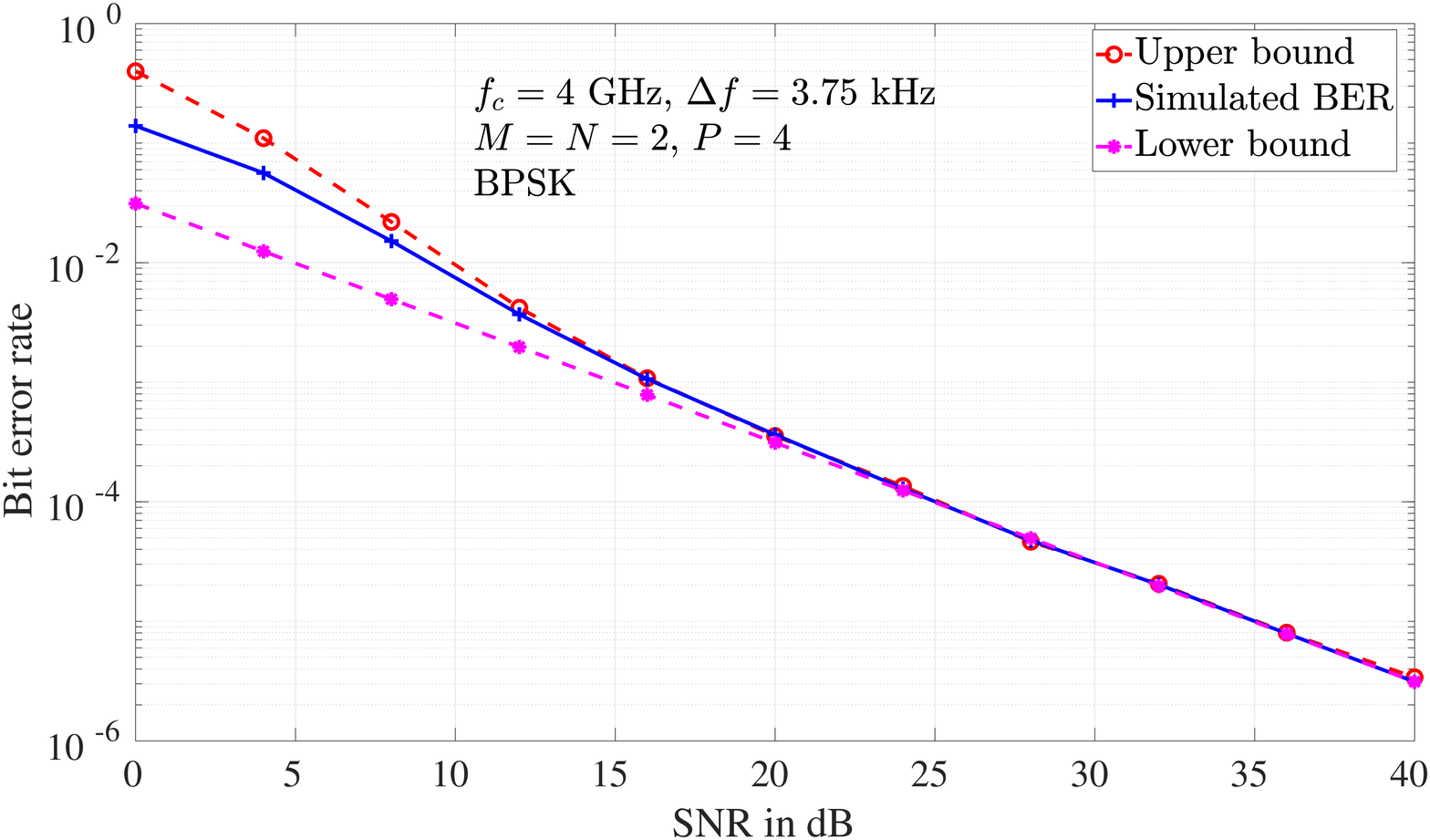}
\vspace{-2mm}
\caption{Upper bound, lower bound, and the simulated BER performance 
of OTFS with $M=N=2$ and $P=4$.}
\vspace{-6mm}
\label{UPLB}
\end{figure}

\begin{figure}
\centering
\includegraphics[width=9.5cm, height=6.5cm]{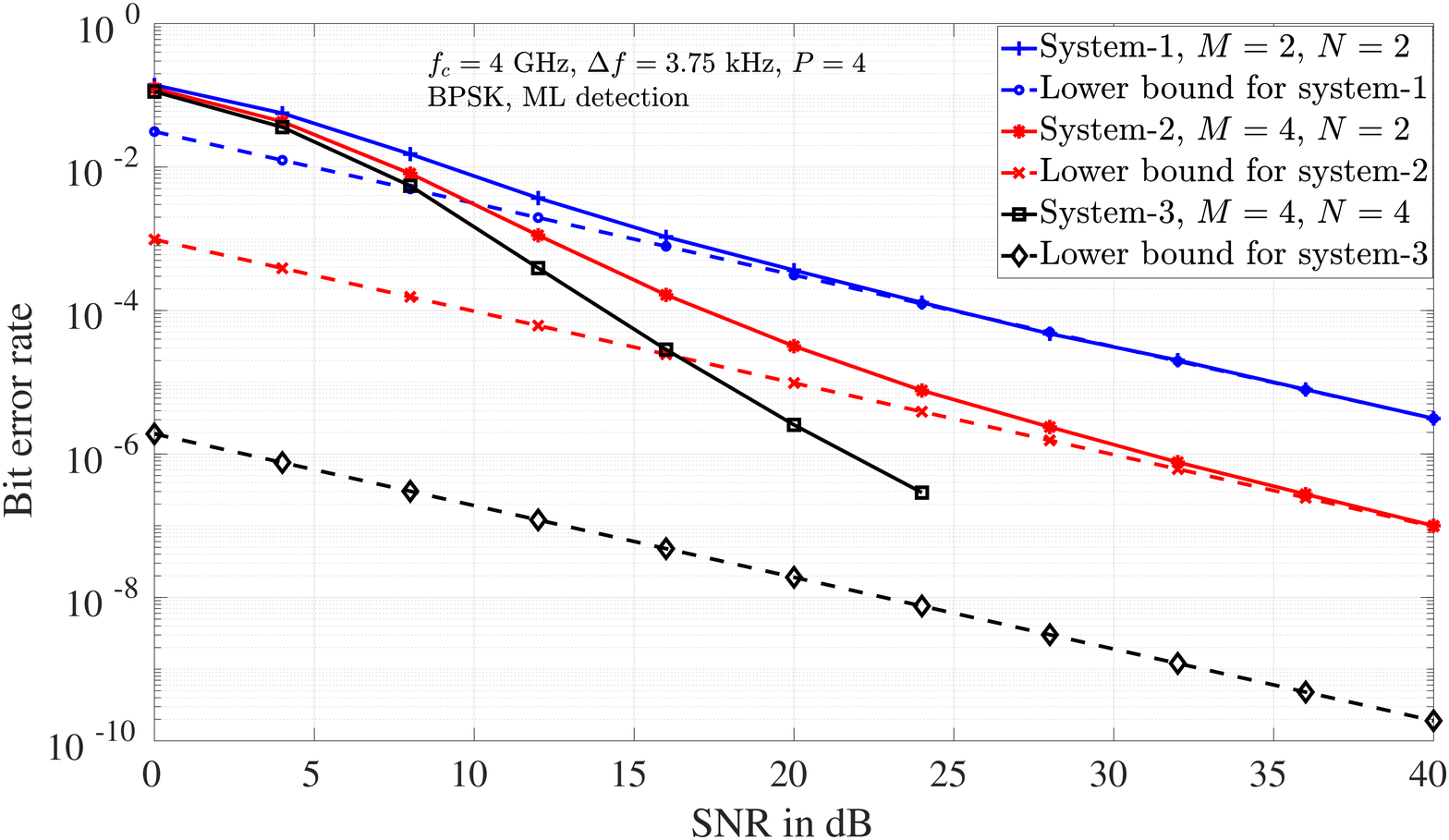}
\vspace{-2mm}
\caption{BER performance of OTFS for $i)$ $M=2$, $N=2$,
$ii)$ $M=4$, $N=2$, $iii)$ $M=4$, $N=4$.}
\vspace{-4mm}
\label{SISO_BER1}
\end{figure}

Figure \ref{SISO_BER1} shows the simulated BER performance of OTFS for 
different values of $M$ and $N$. We consider the following three systems: 
$i)$ system-1 with $M=N=2$, $ii)$ system-2 with $M=4$, $N=2$, and $iii)$ 
system-3 with $M=N=4$. All the three systems use BPSK. 
The maximum Doppler shift considered is $1/NT$. For system-1 and system-2,
which use $N=2$, the maximum Doppler shift is 1.875 kHz, which corresponds 
to a maximum speed of 506.25 km/h. Likewise, the maximum Doppler shift in 
the case of system-3 which uses $N=4$ is 938 kHz, corresponding to a 
maximum speed of 253.125 km/h at 4 GHz carrier frequency. The lower 
bounds given by (\ref{lb3}) for all the three systems are also plotted. 
We note that the $\frac{\kappa}{2^{MN}}$ values for the considered systems 
are $\frac{8}{16}$, $\frac{8}{256}$, and $\frac{8}{65536}$, respectively. 
For illustration purposes, the $\kappa=8$ pairs of matrices 
$(\mathbf{X}_i, \mathbf{X}_j)$ which result in rank one 
$\mathbf{\Delta}_{ij}$ matrices for the system with $M=N=2$ are 
given in Table \ref{error_vec}. Since the $\frac{\kappa}{2^{MN}}$ 
values are decreasing for increasing $MN$, (\ref{lb3}) indicates
that the lower bound for system-3 should lie below that of system-2,
which, in turn, should lie below that of system-1. This trend is clearly 
evident from Fig. \ref{SISO_BER1}. Further, as noted before, for all the 
systems, the BER plots show a diversity order greater than one for low to 
medium SNR values before it meets the diversity one lower bound. An 
interesting observation, however, is that the system-3 with higher
$M$ and $N$ values achieves a higher diversity order compared to those 
of systems-1 and 2 before meeting the lower bound. This is because the 
lower bound for system-3 lies much below the lower bounds for systems-1 
and 2, and the BER curve of system-3 falls with greater slope to meet its
lower bound. This shows that, though the asymptotic diversity order is one,
increasing the value of $MN$ (i.e., increasing the frame size) can lead to 
higher diversity order in the finite SNR regime, resulting in improved
performance for increased frame sizes.  

\begin{figure}
\centering
\includegraphics[width=9.5cm, height=6.5cm]{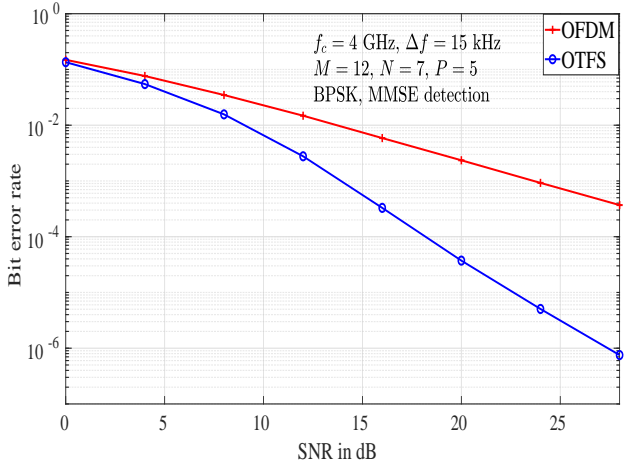}
\vspace{-4mm}
\caption{BER performance comparison of OTFS and OFDM systems 
with MMSE detection for $f_c=4$ GHz, $\Delta f =15$ kHz, $M=12$, $N=7$, 
$P=5$, and BPSK.}
\label{SISO_BER4}
\vspace{-5mm}
\end{figure}

\vspace{-2mm}
\subsection{Results for practical values of $M$ and $N$ in OTFS}
\label{sec3c}
In the previous subsection, we considered small systems with ML detection
to illustrate the asymptotic diversity order of OTFS modulation. We now 
present the performance of OTFS with practical values of $M$ and $N$. 
In Fig. \ref{SISO_BER4}, we present the BER performance of OTFS 
system with $M=12$ and $N=7$ (smallest resource block used in LTE). A 
carrier frequency of 4 GHz, a subcarrier spacing of 15 kHz, exponential 
power delay profile, and Jakes Doppler spectrum \cite{spectrum} are 
considered. Figure \ref{SISO_BER4} also shows the BER performance of 
OFDM system for comparison. Both the systems use minimum mean square 
error (MMSE) detection at the receiver. The maximum Doppler considered 
is 1.85 kHz, which corresponds to a speed of 500 km/h at 4 GHz carrier 
frequency. 
The Doppler shift corresponding the $i$th tap is generated using 
$\nu_i=\nu_{\mbox{\scriptsize{max}}}\cos(\theta_i)$,
where $\nu_{\mbox{\scriptsize{max}}}$ is the maximum Doppler shift and
$\theta_i$ is uniformly distributed over $[-\pi,\pi]$. From the figure, 
it can be seen that the performance of OTFS is significantly superior 
compared to the performance of OFDM. For example, OTFS achieves an SNR 
gain of about 4 dB and 9 dB compared to OFDM at a BER of $10^{-2}$ and 
$10^{-3}$, respectively.

\begin{figure}
\centering
\includegraphics[width=9.5cm, height=6.5cm]{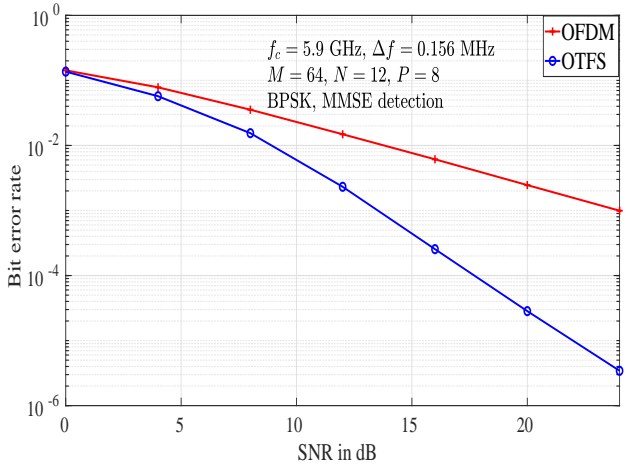}
\vspace{-4mm}
\caption{BER performance comparison of OTFS and OFDM systems 
with MMSE detection for $f_c=5.9$ GHz, $\Delta f =0.156$ MHz, $M=64$, $N=12$, 
$P=8$, and BPSK.}
\vspace{-5mm}
\label{SISO_BER4a}
\end{figure}

Next, in Fig. \ref{SISO_BER4a}, we compare the BER 
performance of OTFS and OFDM considering the system parameters 
according to the IEEE 802.11p standards, which is a standard for 
wireless access in vehicular environments (WAVE) \cite{ieee802.11p}. 
A carrier frequency of 5.9 GHz, a subcarrier frequency of 0.156 MHz,
a frame size $(M,N)=(64,12)$, number of paths $P=8$, a maximum speed 
of 220 km/h, and BPSK are considered. In this WAVE system setting also, 
we observe that the performance of OTFS is significantly better compared 
to that of OFDM. For example, OTFS achieves an SNR gain of about 5 dB and 
10 dB compared to OFDM at a BER of $10^{-2}$ and $10^{-3}$, respectively.
Note that, for the values of $M$, $N$ used in practice
(e.g., $M=12$, $N=7$ in LTE and $M=64$, $N=12$ in IEEE 802.11p) and ML
detection, the transition of the BER slope to diversity one will take
place at very high SNR values.
In coded systems, this uncoded BER performance advantage in OTFS can allow 
the use of high rate codes (e.g., rate 3/4, 7/8) in OTFS systems to achieve 
a given coded BER performance.
A performance comparison 
between OTFS and OFDM in coded settings is presented in Fig. 3 of 
\cite{iter_dec_otfs}, where it is shown that OTFS achieves better
performance compared to OFDM in coded settings as well.

\section{Phase rotation for full diversity in OTFS}
\label{sec4}
In the previous section, we showed that the asymptotic diversity of 
OTFS is one, and that potential for higher diversity orders is observed 
in the finite SNR regime for large frame sizes before the diversity
one regime takes over. In this section, we propose a `phase rotation' 
scheme which extracts the full diversity offered by the delay-Doppler 
channel. From the diversity analysis in Sec. \ref{sec3}, it is clear that 
the asymptotic diversity order of OTFS depends on the minimum rank of the 
difference matrix $\mathbf{\Delta}_{ij} =(\mathbf{X}_i-\mathbf{X}_{j})$,
over all pairs of symbol matrices $(\mathbf{X}_i, \mathbf{X}_j)$. In 
order to design a scheme that can extract full diversity, we take a 
closer look at the symbol matrix $\mathbf{X}$ whose $i$th column is
given by ($\ref{X_mat}$). The symbol matrix $\mathbf{X}$ is a 
$P\times MN$ matrix  that has only $MN$ unique entries, which are 
nothing but the $MN$ symbols of the transmit vector $\mathbf{x}$.
The rows of the matrix $\mathbf{X}$ are the permutations of the 
transmit symbol vector $\mathbf{x}$. When $P=MN$, the matrix 
$\mathbf{X}$ is a block circulant matrix with circulant blocks, as 
shown in (\ref{xstruct}). Note that $\mathbf{X}$ has $M$ circulant 
blocks, each of size $N\times N$, which are cyclically shifted to form a
block circulant matrix. In (\ref{xstruct}), $x_q^{(l)}$ denotes the 
$q$th distinct element of the $l$th block, where $q=0,\cdots,N-1$ and 
$l=0,\cdots,M-1$. When $P<MN$, the $P$ rows of the matrix $\mathbf{X}$
are a subset of rows from (\ref{xstruct}), and the selected subset 
depends on the positions of non-zero entries in the delay-Doppler 
channel matrix. This structure of $\mathbf{X}$ arises naturally from
OTFS pre- and post-processing operations (ISFFT and SFFT) which result 
in the 2D circular convolution of the  transmit vector $\mathbf{x}$ 
with the channel response in delay-Doppler domain.

\begin{figure*}[t]
{\footnotesize 
\begin{equation}
\mathbf{X}=
\left[
\begin{array}{c|c|c|c}
\begin{array}{cccc} x_0^{(0)}& x_1^{(0)} &\cdots &x_{N-1}^{(0)}  \\
          x_{N-1}^{(0)}& x_0^{(0)}&\cdots &x_{N-2}^{(0)}\\
          & &\vdots & \\
          x_1^{(0)}&x_{2}^{(0)}&\cdots & x_0^{(0)}
 \end{array}
   &
\begin{array}{cccc} x_0^{(1)}& 
&\hspace{-6mm} \cdots &x_{N-1}^{(1)}  \\
          x_{N-1}^{(1)}& 
&\hspace{-6mm} \cdots &x_{N-2}^{(1)}\\
          & & \hspace{-6mm} \vdots & \\
          x_1^{(1)}& 
&\hspace{-6mm} \cdots & x_0^{(1)}
 \end{array}
& \cdots 
&
\begin{array}{cccc} x_0^{(M-1)}& 
&\hspace{-6mm} \cdots &x_{N-1}^{(M-1)}  \\
          x_{N-1}^{(M-1)}& 
&\hspace{-6mm} \cdots &x_{N-2}^{(M-1)}\\
          & &\hspace{-6mm} \vdots & \\
          x_1^{(M-1)}& 
&\hspace{-6mm} \cdots & x_0^{(M-1)}
 \end{array}
 \\
\hline
{\scriptsize 
 \begin{array}{cccc} x_0^{(M-1)}& x_1^{(M-1)} &\cdots &x_{N-1}^{(M-1)}  \\
          x_{N-1}^{(M-1)}& x_0^{(M-1)}&\cdots &x_{N-2}^{(M-1)}\\
          & &\vdots & \\
          x_1^{(M-1)}&x_2^{(M-1)}&\cdots & x_0^{(M-1)}
 \end{array}
}
&
\begin{array}{cccc} x_0^{(0)}& 
&\hspace{-6mm} \cdots &x_{N-1}^{(0)}  \\
          x_{N-1}^{(0)}& 
&\hspace{-6mm} \cdots &x_{N-2}^{(0)}\\
          & &\hspace{-6mm} \vdots & \\
          x_1^{(0)}& 
&\hspace{-6mm} \cdots & x_0^{(0)}
\end{array}
&
\cdots 
&
\begin{array}{cccc} x_0^{(M-2)}& 
&\hspace{-6mm} \cdots &x_{N-1}^{(M-2)}  \\
          x_{N-1}^{(M-2)}& 
&\hspace{-6mm} \cdots &x_{N-2}^{(M-2)}\\
          & &\hspace{-6mm} \vdots & \\
          x_1^{(M-2)}& 
&\hspace{-6mm} \cdots & x_0^{(M-2)}
 \end{array}
\\
\hline
 \vdots &  \vdots  & \vdots & \vdots
 \\
 \hline
 \begin{array}{cccc} x_0^{(1)}& x_1^{(1)} &\cdots &x_{N-1}^{(1)}  \\
          x_{N-1}^{(1)}& x_0^{(1)}&\cdots &x_{N-2}^{(1)}\\
          & & \vdots & \\
          x_1^{(1)}&x_2^{(1)}&\cdots & x_0^{(1)}
 \end{array}
&
\begin{array}{cccc} x_0^{(2)}& 
&\hspace{-6mm} \cdots &x_{N-1}^{(2)}  \\
          x_{N-1}^{(2)}& 
&\hspace{-6mm} \cdots &x_{N-2}^{(2)}\\
          & & \hspace{-6mm} \vdots & \\
          x_1^{(2)}& 
&\hspace{-6mm} \cdots & x_0^{(2)}
 \end{array}
& \cdots 
&
\begin{array}{cccc} x_0^{(0)}& 
&\hspace{-6mm} \cdots &x_{N-1}^{(0)}  \\
           x_{N-1}^{(0)}& 
&\hspace{-6mm} \cdots &x_{N-2}^{(0)}\\
           & &\hspace{-6mm} \vdots & \\
           x_1^{(0)}& 
&\hspace{-6mm} \cdots & x_0^{(0)}
 \end{array}
\end{array}
\right].
\label{xstruct}
\end{equation}
}
\vspace{-4mm}
\end{figure*}

The $MN\times 1$ OTFS transmit vector corresponding to symbol matrix 
$\mathbf{X}$ in (\ref{xstruct}) is given by
\vspace{-2mm }
\begin{equation}
\mathbf{x}=[x_0^{(0)},\cdots,x_{N-1}^{(0)},x_0^{(1)},\cdots,x_{N-1}^{(1)}, \cdots,x_{N-1}^{(M-1)}]^T.
\vspace{-2mm}
\label{tx_vec_pr}
\end{equation}
The following theorem shows that multiplying the OTFS transmit vector in
(\ref{tx_vec_pr}) by a diagonal phase rotation matrix $\mathbf{\Phi}$ with 
distinct transcendental numbers results in full diversity.

\begin{theorem}
\label{th1}
Let
\vspace{-2mm}
\begin{equation} 
\hspace{-2mm}
\mathbf{\Phi}= \textnormal{diag}
\left\{ \phi_0^{(0)},\cdots,\phi_{N-1}^{(0)},\phi_0^{(1)},\cdots,\phi_{N-1}^{(1)},\cdots,\phi_{N-1}^{(M-1)}\right\}
\label{phase_rot}
\end{equation} 
be the phase rotation matrix and
\begin{equation}
\mathbf{x}' = \mathbf{\Phi} \mathbf{x} = \begin{bmatrix} \phi_0^{(0)}x_0^{(0)} \\ \vdots \\ \phi_{N-1}^{(0)} x_{N-1}^{(0)} \\  \phi_0^{(1)} x_0^{(1)} \\ \vdots\\ \phi_{N-1}^{(1)} x_{N-1}^{(1)}\\ \vdots \\ \phi_{N-1}^{(M-1)}x_{N-1}^{(M-1)}
\end{bmatrix}
\end{equation} 
be the phase rotated OTFS transmit vector. OTFS with the above phase 
rotation achieves the full diversity of $P$ when $\phi_q^{(l)}=e^{ja_q^{(l)}}$,
$q=0,\cdots,N-1$, $l=0,\cdots, M-1$ are transcendental numbers with 
$a_q^{(l)}$ real, distinct, and algebraic.
\end{theorem}

\begin{proof}
Let $\mathbf{x}'_i=\mathbf{\Phi x}_i$ and $\mathbf{x}'_j=\mathbf{\Phi x}_j$
be two phase rotated OTFS transmit vectors. Let 
$\mathbf{X}'_i$ and $\mathbf{X}'_j$ denote the corresponding phase rotated 
symbol matrices. \\
{\em Case 1: $P=MN$}

When $P=MN$, the symbol matrices $\mathbf{X}'_i$ and $\mathbf{X}'_j$
are block circulant with circulant blocks, and hence
$\mathbf{\Delta}'_{ij}=\mathbf{X}'_i-\mathbf{X}'_j$ also has the
same structure, i.e., 
\begin{equation}
\mathbf{\Delta}'_{ij}=
\left[
\begin{array}{c|c|c|c}
\mathbf{\Delta'}_{ij}^{(0)} & \mathbf{\Delta'}_{ij}^{(1)} & \dots & \mathbf{\Delta'}_{ij}^{(M-1)}\\
\hline
\mathbf{\Delta'}_{ij}^{(M-1)} & \mathbf{\Delta'}_{ij}^{(0)} & \dots & \mathbf{\Delta'}_{ij}^{(M-2)}\\
\hline
\vdots & \vdots & \vdots & \vdots \\
\hline
\mathbf{\Delta'}_{ij}^{(1)} & \mathbf{\Delta'}_{ij}^{(2)} & \dots & \mathbf{\Delta'}_{ij}^{(0)}\\
\end{array}
\right],
\label{Deltaij}
\end{equation}
where
\begin{equation}
\mathbf{\Delta'}_{ij}^{(l)} =
\left[
\begin{array}{cccc}
\delta_0^{(l)}\phi_0^{(l)} & \delta_1^{(l)}\phi_1^{(l)} & \dots & \delta_{N-1}^{(l)}\phi_{N-1}^{(l)}\\
\delta_{N-1}^{(l)}\phi_{N-1}^{(l)} & \delta_0^{(l)}\phi_0^{(l)} & \dots & \delta_{N-2}^{(l)}\phi_{N-2}^{(l)}\\
& \vdots & \vdots &\\
\delta_1^{(l)}\phi_1^{(l)} & \delta_2^{(l)}\phi_2^{(l)} & \dots & \delta_0^{(l)}\phi_0^{(l)}\\
\end{array}
\right],
\end{equation}
where $\delta_q^{(l)}=x_{i,q}^{(l)}-x_{j,q}^{(l)}$, with $x_{i,q}^{(l)}$ 
and $x_{j,q}^{(l)}$ being $q$th distinct elements in the $l$th block of 
$\mathbf{X}_i$ and $\mathbf{X}_j$, respectively. Since 
$\mathbf{\Delta'}_{ij}$ is block circulant with circulant blocks, it is 
diagonalized by $\mathbf{F}_M \otimes \mathbf{F}_N$, where $\mathbf{F}_M$
and $\mathbf{F}_N$ denote the $M \times M$ and $N \times N$ DFT matrices 
and $\otimes$  denotes the Kronecker product. Therefore, 
$\mathbf{\Delta'}_{ij}$ is given by \cite{davis}
\begin{equation}
\mathbf{\Delta'}_{ij} = (\mathbf{F}_M \otimes \mathbf{F}
_N)^H\mathbf{D}(\mathbf{F}_M \otimes \mathbf{F}_N),
\end{equation}
where $\mathbf{D}$ is an $MN \times MN $ diagonal matrix whose entries 
are eigen values of $\mathbf{\Delta'}_{ij}$, given by
\vspace{-2mm}
\begin{equation}
\mathbf{D} = \sum_{l=0}^{M-1} \mathbf{\Omega}_M^{l} \otimes \mathbf{\Lambda}^{(l)},
\label{diagD}
\end{equation}
where $\mathbf{\Omega}_M=\mbox{diag}\{1,\omega,\omega ^2,\cdots,\omega^{M-1}\}$
with $\omega = e^{j2\pi/M}$, and $\mathbf{\Lambda}^{(l)}$ is a diagonal matrix 
whose entries are the eigen values of $\mathbf{\Delta'}_{ij}^{(l)}$.
Let $\lambda_q^{(l)}$ denote the $q$th eigen value of 
$\mathbf{\Delta'}_{ij}^{(l)}$ and $\mu_0,\mu_1,\cdots,\mu_{MN-1}$ denote 
the eigen values of $\mathbf{\Delta'}_{ij}$. From (\ref{diagD}), the $k$th 
eigen value of $\mathbf{\Delta'}_{ij}$, that is, $\mu_k$ is given by
\begin{equation}
\mu_k = \sum_{l=0}^{M-1} \lambda_u^{(l)} \omega^{vl},
\label{mu_k}
\end{equation}
where $u=k-\lfloor \frac{k}{N} \rfloor N$, 
$v=\lfloor \frac{k}{N} \rfloor$, and $\lambda_u^{(l)}$ is the $u$th 
eigen value of $\mathbf{\Delta'}_{ij}^{(l)}$, given by
\begin{equation}
\lambda_u^{(l)} = \sum_{q=0}^{N-1} \phi_q^{(l)} \delta_q^{(l)}e^{-j2\pi uq/N }.
\label{lambda_ul}
\end{equation}
Using (\ref{lambda_ul}) in (\ref{mu_k}), we have
\begin{eqnarray}
\mu_k &= &\sum_{l=0}^{M-1} \sum_{q=0}^{N-1} \phi_q^{(l)} \delta_q^{(l)}e^{-j2\pi uq/N } \omega^{vl} \nonumber \\
&= &\sum_{l=0}^{M-1} \bigg\{ \phi_0^{(l)}\delta_0^{(l)} + \phi_1^{(l)}\delta_1^{(l)}e^{-j2\pi u/N} +  \nonumber \\
& & 
\cdots + \phi_{N-1}^{(l)}\delta_{N-1}^{(l)}e^{-j2\pi u(N-1)/N} \bigg\} \omega^{vl},
\end{eqnarray}
which can be further simplified as
\begin{eqnarray}
\mu_k & \hspace{-3mm} = \hspace{-3mm} &  \phi_0^{(0)}\delta_0^{(0)} + \phi_1^{(0)}\delta_1^{(0)}e^{-\frac{j2\pi u}{N}} + \nonumber \\
& & \cdots + \phi_{N-1}^{(0)} \delta_{N-1}^{(0)}e^{-\frac{j2\pi u(N-1)}{N} } + \ \phi_0^{(1)}\delta_0^{(1)}\omega^v  + \nonumber \\
& & \cdots +  \phi_{N-1}^{(1)}\delta_{N-1}^{(1)}e^{-\frac{j2\pi u(N-1)}{N}} \omega^v + \nonumber \\\
& & \cdots + \phi_0^{(M-1)}\delta_0^{(M-1)}\omega^{v(M-1)}  + \nonumber \\
& & \cdots + \ \phi_{N-1}^{(M-1)}\delta_{N-1}^{(M-1)}e^{-\frac{j2\pi u(N-1)}{N}} \omega^{v(M-1)}.
\label{expansion}
\end{eqnarray}
At this stage, we invoke the Lindenmann's theorem \cite{stbc_numthry}, 
which states that, if $a_1,a_2,\cdots,a_m$ are distinct algebraic numbers, 
and if $c_1,c_2,\cdots,c_m$ are algebraic and not all equal to zero, then
\vspace{-1mm}
\begin{equation}
c_1e^{a_1} + c_2e^{a_2}+ \cdots + c_m e^{a_m} \neq 0.
\label{lind}
\vspace{-1mm}
\end{equation}
It should be noted from (\ref{expansion}) that the terms of the form 
$\delta_q^{(l)}e^{-\frac{j2\pi uq}{N}} \omega^{vl}$ are all algebraic 
\cite{stbc_numthry}. Therefore, comparing (\ref{lind}) and 
(\ref{expansion}), if the terms  $\phi_q^{(l)}$ are chosen such that 
$\phi_q^{(l)}=e^{ja_q^{(l)}}$ are transcendental with $a_q^{(l)}$ 
real, distinct, and algebraic, then $\mu_k$ can not be zero. Since $\mu_k$s 
are eigen values of $\mathbf{\Delta'}_{ij}$, choosing the diagonal phase 
rotation matrix with entries $\phi_q^{(l)}=e^{ja_q^{(l)}}$ being 
transcendental with $a_q^{(l)}$ real, distinct, and algebraic 
ensures that all the eigen values of $\mathbf{\Delta'}_{ij}$ are non-zero,
making it full rank (i.e., rank $P$). Since this is true for 
$\mathbf{\Delta'}_{ij}$ for all $(i,j)$, $i\neq j$, the minimum rank 
of $\mathbf{\Delta'}_{ij}$ is equal to $MN$. Hence, from (\ref{div_order}), 
the achieved diversity order of OTFS with the proposed phase rotation is 
$MN$. \\
{\em Case 2: $P<MN$} 

Now, consider the case when $P<MN$. As mentioned previously, when $P<MN$,
the rows of the transmit symbol matrix $\mathbf{X}$ is the subset of rows 
from the corresponding $MN \times MN $ matrix in (\ref{xstruct}). If 
$\mathbf{X}'_i$ and $\mathbf{X}'_j$ are two phase rotated symbol matrices 
with $P<MN$, then the rows of 
$\mathbf{\Delta}'_{ij}=\mathbf{X}'_i-\mathbf{X}'_j$ form a subset of the 
rows of the corresponding $MN\times MN$ matrix in (\ref{Deltaij}). Since 
the matrix  in (\ref{Deltaij}) is shown to be full rank in {\em Case 1}, 
$\mathbf{\Delta}'_{ij}$ with $P<MN$ should have a rank equal to $P$. 
Therefore, OTFS with phase rotation using transcendental numbers
of the form $\phi_q^{(l)}=e^{ja_q^{(l)}}$ with $a_q^{(l)}$ 
real, distinct, and algebraic achieves the full diversity of $P$ in the 
delay-Doppler domain.
\end{proof}

\vspace{-2mm}
\subsection{Simulation results}
\label{sec4a}
Figure \ref{SISO_BER5} shows the simulated BER performance of OTFS without 
and with phase rotation for $i)$ system-1 with $M=N=2$, $ii)$ system-2 with 
$M=4$, $N=2$, and $iii)$ system-3 with $M=N=4$. The carrier frequency and 
the subcarrier spacing used are 4 GHz and 3.75 kHz, respectively. All the 
systems use BPSK. Other simulation parameters are as given in Table 
\ref{SimPar}. For the simulations, all the three systems use the phase 
rotation matrix, 
$\mathbf{\Phi}= \ \textnormal{diag}\{{ 1,e^{j\frac{1}{MN}} \cdots e^{j\frac{MN-1}{MN}} }\}$. 
From Fig. \ref{SISO_BER5}, we observe that the asymptotic diversity order 
of all the three systems without phase rotation is one. Further, the OTFS
systems with phase rotation exhibit full diversity in the high SNR regime. 
Although all the systems with phase rotation exhibit a diversity order of 
$P=4$, we observe a slight difference in the BER performance 
of the three systems. This is because of the different coding gains achieved 
by each system. While the proposed phase rotation scheme achieves full 
delay-Doppler diversity, the coding gain achieved by a system can be 
improved by optimizing the phases used in the phase rotation matrix 
\cite{stbc_numthry}. 

\begin{figure}
\centering
\includegraphics[width=9.5cm, height=6.5cm]{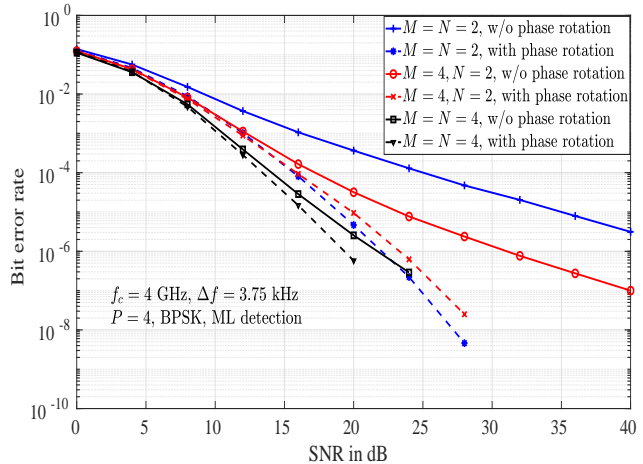}
\vspace{-2mm}
\caption{BER performance of OTFS without and with phase 
rotation for $i)$ $M=N=2$, $ii)$ $M=4$, $N=2$, and $iii)$ $M=N=4$, 
and BPSK.}
\vspace{-4mm}
\label{SISO_BER5}
\end{figure}

In Fig. \ref{SISO_BER5x}, we present the simulated BER performance of OTFS 
with and without phase rotation for a system with $M=N=2$ and 8-QAM. The 
carrier frequency and the subcarrier spacing used are 4 GHz and 3.75
kHz, respectively. Other simulation parameters are as given in Table 
\ref{SimPar}.  For the simulations, the phase rotation matrix, 
$\mathbf{\Phi}=\textnormal{diag} \{{ 1,e^{j\frac{1}{MN}} \cdots e^{j\frac{MN-1}{MN}} }\}$ 
is used. From Fig. \ref{SISO_BER5x}, we observe that OTFS without phase 
rotation achieves a diversity order of one. Whereas, OTFS with phase rotation 
shows the intended diversity benefit. For example, at a BER of $10^{-5}$, 
the OTFS system with phase rotation achieves an SNR gain of about 17 dB 
compared to the system without phase rotation.

\begin{figure}
\centering
\includegraphics[width=9.5cm, height=6.5cm]{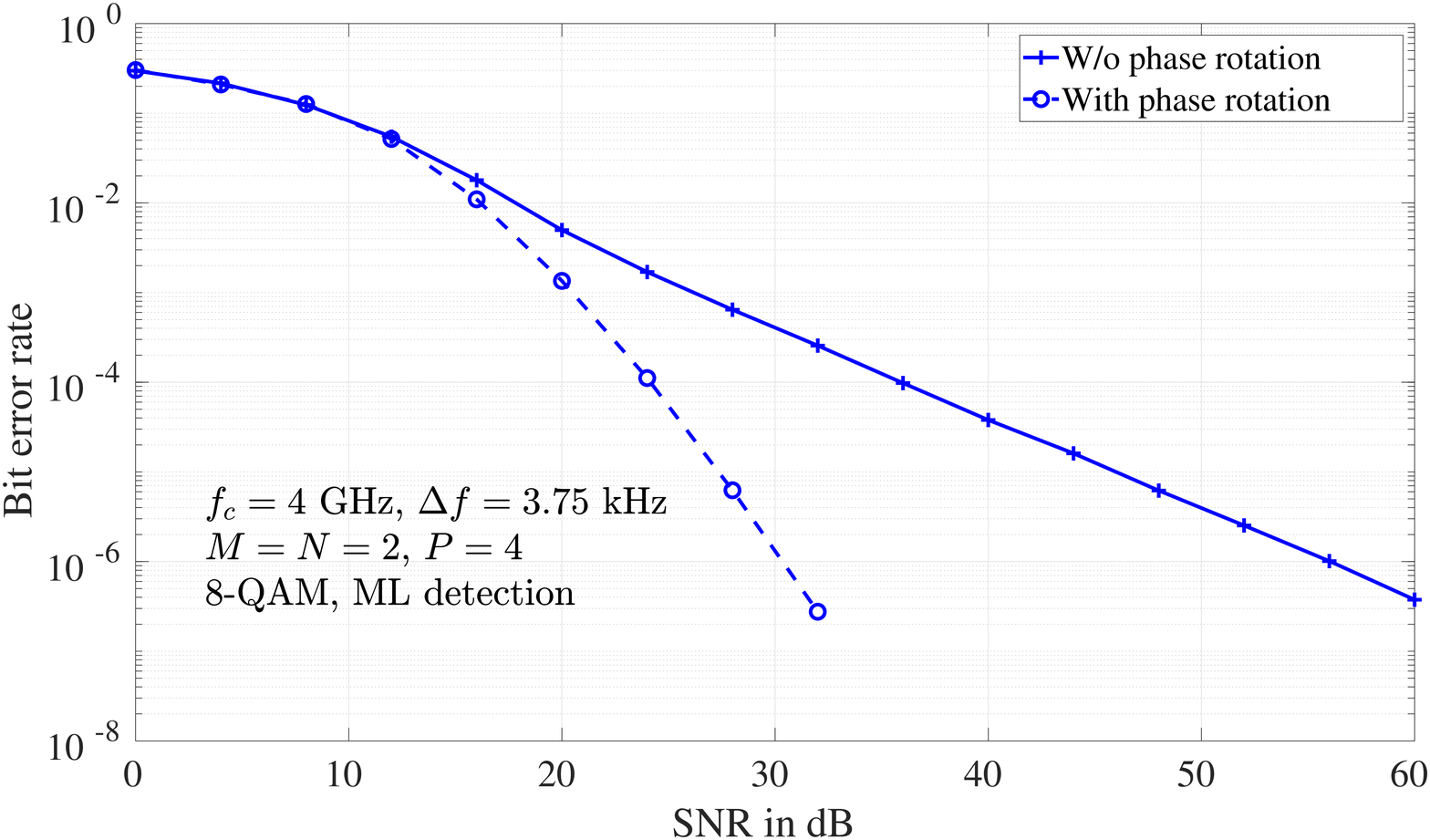}
\vspace{-2mm}
\caption{BER performance of OTFS without and with phase rotation, 
$M=N=2$, and 8-QAM.}
\vspace{-4mm}
\label{SISO_BER5x}
\end{figure}

\section{MIMO-OTFS modulation}
\label{sec5}
In this section, we consider OTFS modulation and its diversity order
in a MIMO setting. 

\subsection{MIMO-OTFS system model}
\label{sec5a}
Consider a MIMO-OTFS system shown in Fig. \ref{MIMO-OTFS} with $n_t$ 
transmit  and $n_r$ receive antennas. Each antenna transmits an 
independent OTFS signal vector. The channel gain between the $k$th 
transmit antenna and $l$th receive antenna in the delay-Doppler domain 
corresponding to delay $\tau$ and Doppler $\nu$ is given by
\begin{equation}
h_{lk}(\tau,\nu)=\sum_{i=1}^{P}h_{lk_i}
\delta(\tau -\tau_i) \delta(\nu-\nu_i),
\label{sparsechannelmimo}
\end{equation}
$k=1,2,\cdots,n_t$, $l=1,2,\cdots,n_r$, and $P$ is the number of 
channel taps. Let $\mathbf{H}_{lk}$ denote the $MN\times MN$ equivalent
channel matrix between the $k$th transmit antenna and $l$th receive
antenna. Let $\mathbf{x}_k$ denote the $MN\times 1$ transmit vector 
from the $k$th transmit antenna and $\mathbf{y}_l$ denote the 
$MN\times 1$ received vector at the $l$th receive antenna. Then, 
using the linear vector channel model in (\ref{vecform}) for SISO-OTFS, 
the linear system model describing the input and output relation in
MIMO-OTFS can be obtained as
\begin{eqnarray}
\mathbf{y}_1 & = & \mathbf{H}_{11}\mathbf{x}_1+\mathbf{H}
_{12}\mathbf{x}_2 + \cdots + \mathbf{H}_{1{n_t}}\mathbf{x}
_{n_t}+\mathbf{v}_1,   \nonumber \\
\mathbf{y}_2 & = & \mathbf{H}_{21}\mathbf{x}_1+\mathbf{H}
_{22}\mathbf{x}_2 + \cdots + \mathbf{H}_{2{n_t}}\mathbf{x}
_{n_t}+\mathbf{v}_2, \label{mimoeqns} \\
\vdots  \nonumber\\
\hspace{-4mm}\mathbf{y}_{n_r} & = & \mathbf{H}_{{n_r}1}
\mathbf{x}_1+\mathbf{H}_{{n_r}2}\mathbf{x}_2 + \cdots +
\mathbf{H}_{{n_r}{n_t}}\mathbf{x}_{n_t}+\mathbf{v}_{n_r}. \nonumber
\hspace{4mm}
\end{eqnarray}
Defining 
\begin{align*}
\  \mathbf{H}_{ {\tiny \mbox{MIMO}}} &= \begin{bmatrix}
\mathbf{H} _{11} & \mathbf{H}_{12} & \dots & \mathbf{H}
_{1{n_t}} \\
\mathbf{H} _{21} & \mathbf{H}_{22} & \dots & \mathbf{H}
_{2{n_t}} \\
\vdots & \vdots & \ddots & \vdots \\
\mathbf{H} _{{n_r}1} & \mathbf{H}_{{n_r}2} & \dots &
\mathbf{H}_{{n_r}{n_t}}
\end{bmatrix},
\end{align*}

\hspace{-5mm}
\begin{small}
$\mathbf{x}_{{\tiny \mbox{MIMO}}}={[{\mathbf{x}_1}^{T},
{\mathbf{x}_2}^{T},\cdots,{\mathbf{x}_{n_t}}^{T}] }^{T}$,
$\mathbf{y}_{{\tiny \mbox{MIMO}}} = {[{\mathbf{y}_1}^{T},
{\mathbf{y}_2}^{T},\cdots, {\mathbf{y}_{n_r}}^{T}] }^{T}$, 
$\mathbf{v}_{{\tiny \mbox{MIMO}}} = {[{\mathbf{v}_1}^{T},
{\mathbf{v}_2}^{T},\cdots,  {\mathbf{v}_{n_r}}^{T}] }^{T}$,
\end{small}
(\ref{mimoeqns}) can be written as
\begin{equation}
\mathbf{y}_{{\tiny \mbox{MIMO}}} = \mathbf{H}_{{\tiny
\mbox{MIMO}}}\mathbf{x}_{{\tiny \mbox{MIMO}}} + \mathbf{v}
_{\tiny \mbox{MIMO}},
\label{mimovecform}
\end{equation}
where
$\mathbf{x}_{{\tiny \mbox{MIMO}}} \in \mathbb{C} ^{{n_t}MN
\times 1} ,\mathbf{y}_{{\tiny \mbox{MIMO}}}, \mathbf{v}
_{\tiny \mbox{MIMO}} \in \mathbb{C} ^{{n_r}MN \times 1}$, and
$\mathbf{H}_{{\tiny \mbox{MIMO}}} \in \mathbb{C}^{{n_r}MN
\times {n_t}MN}$.

\begin{figure*}[t]
\centering
\includegraphics[width=14 cm, height=4.0cm]{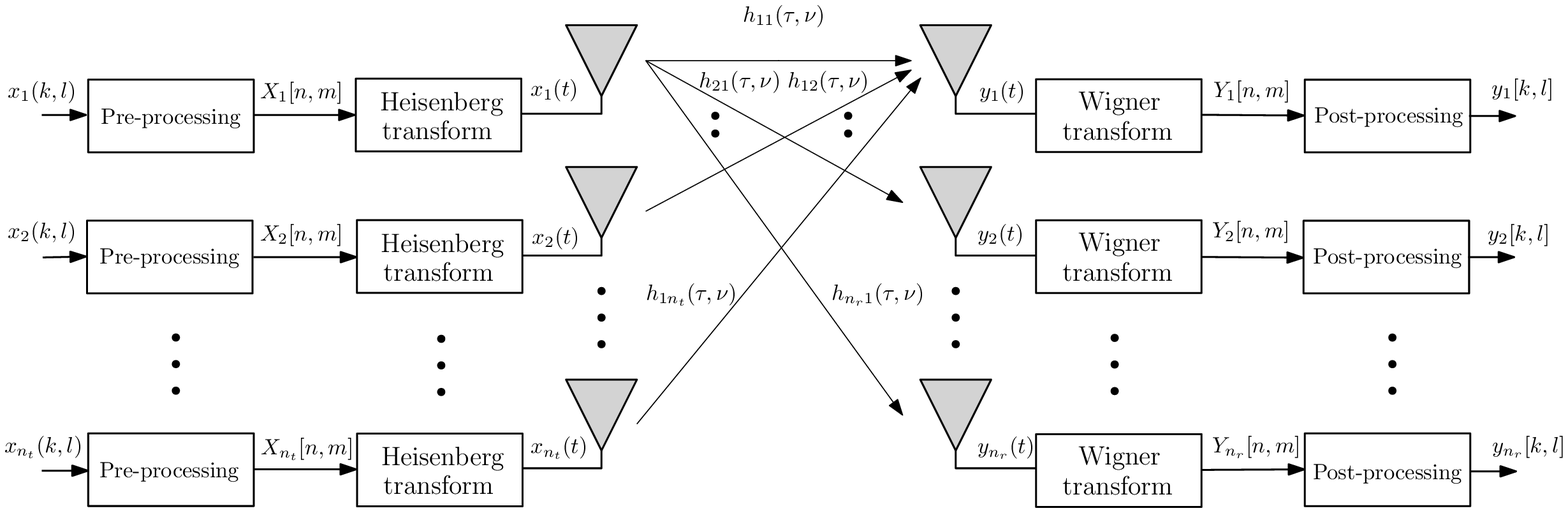}
\vspace{2mm}
\caption{MIMO-OTFS modulation scheme.}
\label{MIMO-OTFS}
\vspace{-4mm}
\end{figure*}

\subsection{Diversity of MIMO-OTFS}
\label{sec5b}
In this subsection, we derive the asymptotic diversity order of 
MIMO-OTFS. For this, first note that, in the effective channel
matrix $\mathbf{H}_{\tiny \mbox{MIMO}}$, each $\mathbf{H}_{lk}$ has
only $P$ unique entries, and hence $\mathbf{H}_{\tiny \mbox{MIMO}}$
has $Pn_tn_r$ unique entries. Further, each row of 
$\mathbf{H}_{{\tiny \mbox{MIMO}}}$ has only $n_tP$ non-zero elements 
and each column has only $n_rP$ non-zero elements. Following 
(\ref{hXform}), the MIMO-OTFS system model in (\ref{mimovecform}) 
can be written as
\begin{equation}
\begin{bmatrix}
\mathbf{y}_1^T \\
\mathbf{y}_2^T \\
\vdots\\
\mathbf{y}_{n_r}^T
\end{bmatrix}
 =
\begin{bmatrix}
\mathbf{h}'_{11} & \mathbf{h}'_{12} \cdots \mathbf{h}'_{1n_t} \\
\mathbf{h}'_{21} & \mathbf{h}'_{22} \cdots \mathbf{h}'_{2n_t} \\
\vdots \\
\mathbf{h}'_{n_r1} & \mathbf{h}'_{n_r2} \cdots \mathbf{h}'_{n_rn_t} \\
\end{bmatrix}
\begin{bmatrix}
\mathbf{X}_{1} \\
\mathbf{X}_{2}\\
\vdots\\
\mathbf{X}_{n_t}
\end{bmatrix}
+
\begin{bmatrix}
\mathbf{v}_1^T \\
\mathbf{v}_2^T \\
\vdots\\
\mathbf{v}_{n_r}^T
\end{bmatrix},
\label{hX_MIMO1}
\end{equation}
or equivalently
\begin{equation}
{\tilde{\bf y}}={\tilde{\bf H}}{\tilde{\bf X}}+\tilde{\mathbf{V}},
\label{hX_MIMO}
\end{equation}
where ${\tilde{\bf y}}$ is an $n_r\times MN$ received signal matrix 
whose $l$th row is the received vector received in the $l$th receive 
antenna, ${\tilde{\bf X}}$ is an $n_tP\times MN$ matrix obtained by 
stacking $n_t$ number of $P\times MN$ sized symbol matrices of the 
form (\ref{X_mat}), ${\tilde{\bf H}} \in \mathbb{C}^{n_r\times n_tP}$ 
is the channel matrix with $\mathbf{h}'_{lk} \in \mathbb{C}^{1\times P}$ 
consisting of $P$ unique non-zero entries of $\mathbf{H}_{lk}$, and
${\tilde{\bf V}} \in \mathbb{C}^{n_r\times MN}$ is the noise matrix.

Let $\tilde{\mathbf{X}}_i$ and $\tilde{\mathbf{X}}_j$ be two symbol 
matrices.  Assuming perfect channel state information and ML detection 
at the receiver, the probability of decoding the transmitted symbol 
matrix $\tilde{\mathbf{X}}_i$ in favor of $\tilde{\mathbf{X}}_j$ is 
given by
\begin{equation}
P(\tilde{\mathbf{X}}_{i}\rightarrow \tilde{\mathbf{X}}_j|\tilde{\mathbf{H}})=
Q \left( \sqrt{\frac{\|\tilde{\mathbf{H}}
(\tilde{\mathbf{X}}_i-\tilde{\mathbf{X}}_j)\|^2}{2N_0}} \right),
\label{MIMO_PEP1}
\end{equation}
and the average PEP is given by
\begin{equation}
P(\tilde{\mathbf{X}}_{i}\rightarrow \tilde{\mathbf{X}}_{j})=
\mathbb{E} \left[Q \left( \sqrt{\frac{\|\tilde{\mathbf{H}}
(\tilde{\mathbf{X}}_i-\tilde{\mathbf{X}}_j)\|^2}{2N_0}} \right)\right].
\label{MIMO_PEP2}
\end{equation}
Using Chernoff bound and the fact that each antenna transmits independent 
OTFS symbols, an upper bound on the PEP in (\ref{MIMO_PEP2}) can be 
obtained as \cite{DTse}
\begin{equation}
P(\tilde{\mathbf{X}}_{i}\rightarrow \tilde{\mathbf{X}}_{j}) \leq \left( \prod \limits_{l=1}^{r}\frac{1}{1+\frac{\gamma\lambda_{k,l}^2}{4P} }\right )^{n_r},
\label{MIMO_PEP3}
\end{equation}
where $\gamma=\frac{1}{N_0}$ is the SNR per receive antenna, $\lambda_{k,l}$ 
is the $l$th singular value of the difference matrix 
$\mathbf{\Delta}_{k,ij}=(\mathbf{X}_{k,i}-\mathbf{X}_{k,j})$ with 
$\mathbf{X}_{k,i}$ and $\mathbf{X}_{k,j}$ denoting OTFS symbol matrices 
transmitted from $k$th antenna (for some $k\in {1,2,\cdots,n_t}$) in 
$\tilde{\mathbf{X}}_i$ and $\tilde{\mathbf{X}}_j$, respectively,
and $r$ is the rank of $\mathbf{\Delta}_{k,ij}$. At high SNR values, 
(\ref{MIMO_PEP3}) simplifies to
\begin{equation}
P(\tilde{\mathbf{X}}_i\rightarrow \tilde{\mathbf{X}}_j)\leq
\frac{1}{\gamma^{n_rr}}\left(\prod_{l=1}^{r}
\frac{\lambda_{k,l}^2}{4P}\right)^{-n_r}.
\end{equation}
The PEP term with the minimum value of $r$ dominates the overall BER.  
Therefore, the diversity order achieved by MIMO-OTFS, denoted by 
$\rho_{\tiny \mbox{mimo-otfs}}$ is given by 
\begin{equation}
\rho_{\tiny \mbox{mimo-otfs}} =n_r\cdot \min_{i,j \ i\neq j}\ \mbox{rank}(\mathbf{\Delta}_{k,ij}). 
\label{div_order_mimo}
\end{equation}
Now, similar to the case of SISO-OTFS, if 
$\mathbf{X}_{k,i}=a.\mathbf{1}_{P\times MN}$ and
$\mathbf{X}_{k,j}=a'.\mathbf{1}_{P\times MN}$, then the 
difference matrix $\mathbf{\Delta}_{k,ij}=(a-a').\mathbf{1}_{P\times MN}$
has rank one. Hence, the asymptotic diversity order of MIMO-OTFS is $n_r$.
\hspace{4.4cm} $\square$

\subsection{Phase rotation for full diversity in MIMO-OTFS}
\label{sec5d}
In this subsection, we consider phase rotation to extract the full
diversity in MIMO-OTFS. The transmit vector in MIMO-OTFS is a 
concatenation of $n_t$ independent OTFS transmit vectors of size 
$MN\times 1$ as described by (\ref{mimovecform}). The $MN\times 1$
OTFS transmit vector from each antenna is multiplied by the phase 
rotation matrix $\mathbf{\Phi}$ given in (\ref{phase_rot}). The phase 
rotated MIMO-OTFS transmit vector is then given by
\begin{equation}
\mathbf{x}'_{ \tiny \mbox{{MIMO}}}= (\mathbf{I}_{n_t} \otimes \mathbf{\Phi})\mathbf{x}_{\tiny \mbox{{MIMO}}}.
\end{equation}
Let $\tilde{\mathbf{X}}'$ be the phase rotated MIMO-OTFS
symbol matrix corresponding to $\mathbf{x}'$. From
(\ref{hX_MIMO1}), $\tilde{\mathbf{X}}'$ is of the form
\begin{equation}
\tilde{\mathbf{X}}' = \begin{bmatrix}
\mathbf{X}'_{1} \\
\mathbf{X}'_{2}\\
\vdots\\
\mathbf{X}'_{n_t}
\end{bmatrix},
\end{equation}
where $\mathbf{X}'_k$ is the phase rotated OTFS symbol matrix
corresponding to the $k$th transmit antenna.
If $\tilde{\mathbf{X}}'_i$ and $\tilde{\mathbf{X}}'_j$ are two phase
rotated MIMO-OTFS symbol matrices corresponding to the
transmit vectors $\mathbf{x}'_{i, \tiny{\mbox{MIMO}}}$  and
$\mathbf{x}'_{j, \tiny{\mbox{MIMO}}}$, then their
difference matrix $\tilde{\mathbf{\Delta}}'_{ij}$ is of the form
\begin{equation}
\tilde{\mathbf{\Delta}}'_{ij} = \begin{bmatrix}
\mathbf{\Delta}'_{1,ij} \\
\mathbf{\Delta}'_{2,ij}\\
\vdots\\
\mathbf{\Delta}'_{n_t, ij}
\end{bmatrix},
\label{del_til}
\end{equation}
where $\mathbf{\Delta}'_{k,ij} = \mathbf{X}'_{k,i}-\mathbf{X}'_{k,j}$, 
with $\mathbf{X}'_{k,i}$ and $\mathbf{X}'_{k,j}$ being the phase 
rotated OTFS symbol matrices corresponding to the $k$th antenna in 
$\tilde{\mathbf{X}}'_i$ and $\tilde{\mathbf{X}}'_j$, respectively. 
From Sec. \ref{sec4}, it is known that $\mathbf{\Delta}'_{k,ij}$ has 
rank equal to $P$ for all $k={1,2,\cdots,n_t}$. Using this fact 
in (\ref{div_order_mimo}), the diversity order achieved by phase 
rotated MIMO-OTFS system is equal to $Pn_r$. \hspace{76mm} $\square$

\subsection{Simulation results}
\label{sec5c}
Figure \ref{MIMO_BER1} shows the BER performance of $1\times 1$ 
SISO-OTFS and $2\times 2$ MIMO-OTFS systems. Both the systems use 
$M=N=2$ and BPSK. The number of channel taps considered is $P=4$. The 
carrier frequency and the subcarrier spacing used are 4 GHz and 3.75 kHz, 
respectively. The considered simulation  parameters are  summarized in 
Table \ref{SimPar}. From the figure, it is observed that the
simulated BER for $1\times1$ SISO-OTFS and $2\times 2$ MIMO-OTFS 
show diversity orders of one and two, respectively, verifying the 
analytical diversity order derived in the previous subsection. 

\begin{figure}
\centering
\includegraphics[width=9.5cm, height=6.5cm]{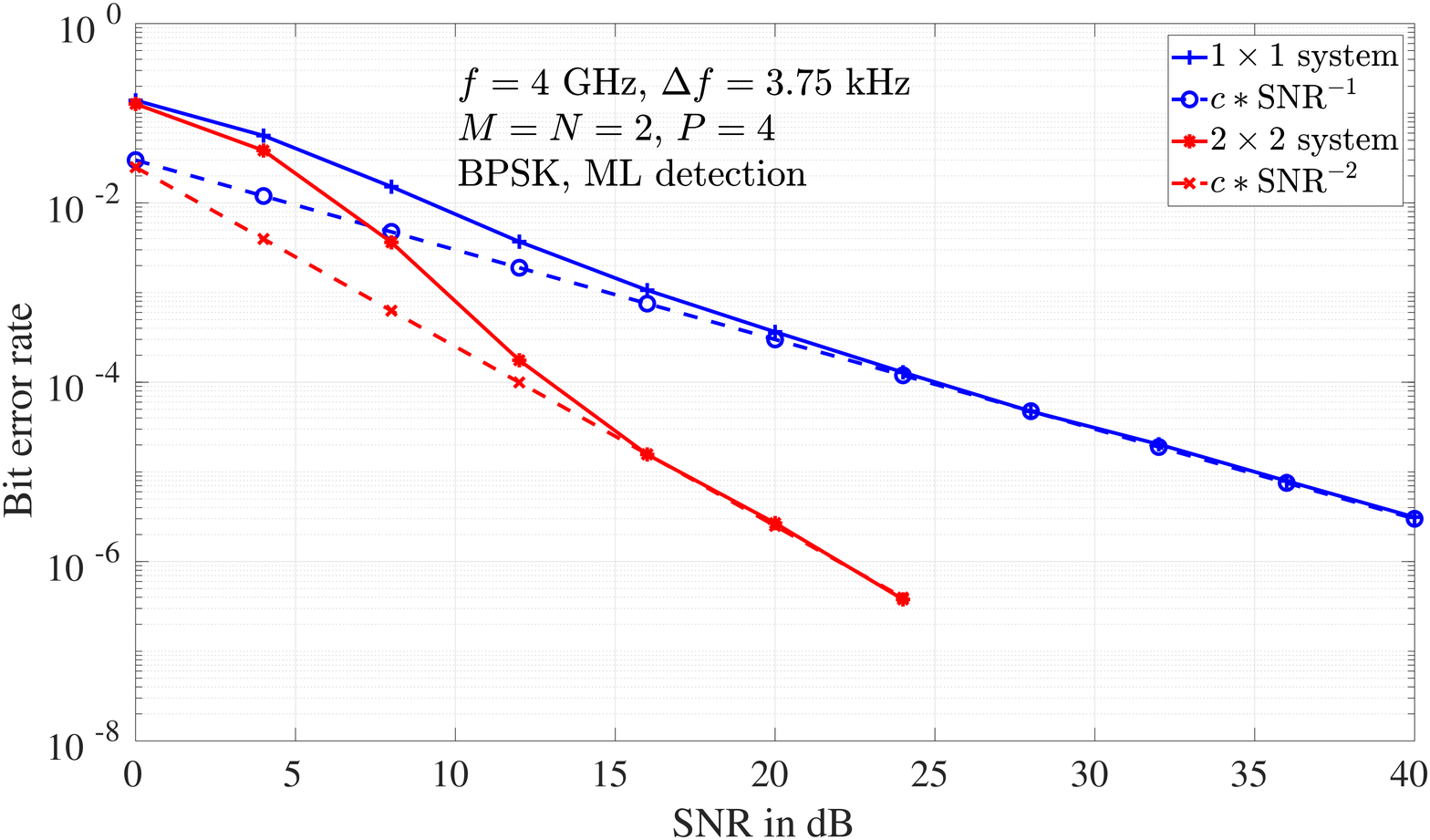}
\vspace{-2mm}
\caption{BER performance of $1\times 1$ SISO-OTFS and $2\times 2$ MIMO-OTFS 
systems.}
\vspace{-4mm}
\label{MIMO_BER1}
\end{figure}

Next, we consider the effect of increasing the frame size (i.e., $MN$) on 
the BER performance. Figure \ref{MIMO_BER2} shows the BER performance of
$1\times 2$ system with $i)$ $M=N=2$ and $ii)$ $M=4$, $N=2$. Both the 
systems use BPSK. The number of channel taps considered is $P=4$.  
Other simulation parameters are as given in Table \ref{SimPar}.
From the figure, we observe that the BER performance of the system 
with $M=4$ and $N=2$ is better than the system with $M=N=2$. This is 
similar to the SISO-OTFS result shown in Sec. \ref{sec3b}. Specifically, 
increasing the frame size ($MN$) results in  higher diversity order
in the finite SNR regime, before the asymptotic diversity order of 
$\rho_{\tiny \mbox{mimo-otfs}}=2$ takes over. It can be observed that, 
MIMO-OTFS can achieve diversity orders closer to $Pn_r$ in the finite 
SNR regime, as the size of the OTFS frame $MN$ is increased.
Figure \ref{PR_MIMO_OTFS} shows the BER performance of a $2\times 2$ 
MIMO-OTFS systems with and without phase rotation with $M=N=2$ and
BPSK. From the figure, it can be seen that the MIMO-OTFS system with 
phase rotation achieves the intended diversity benefit compared
to the diversity in MIMO-OTFS without phase rotation. 

\begin{figure}
\centering
\includegraphics[width=9.5cm, height=6.5cm]{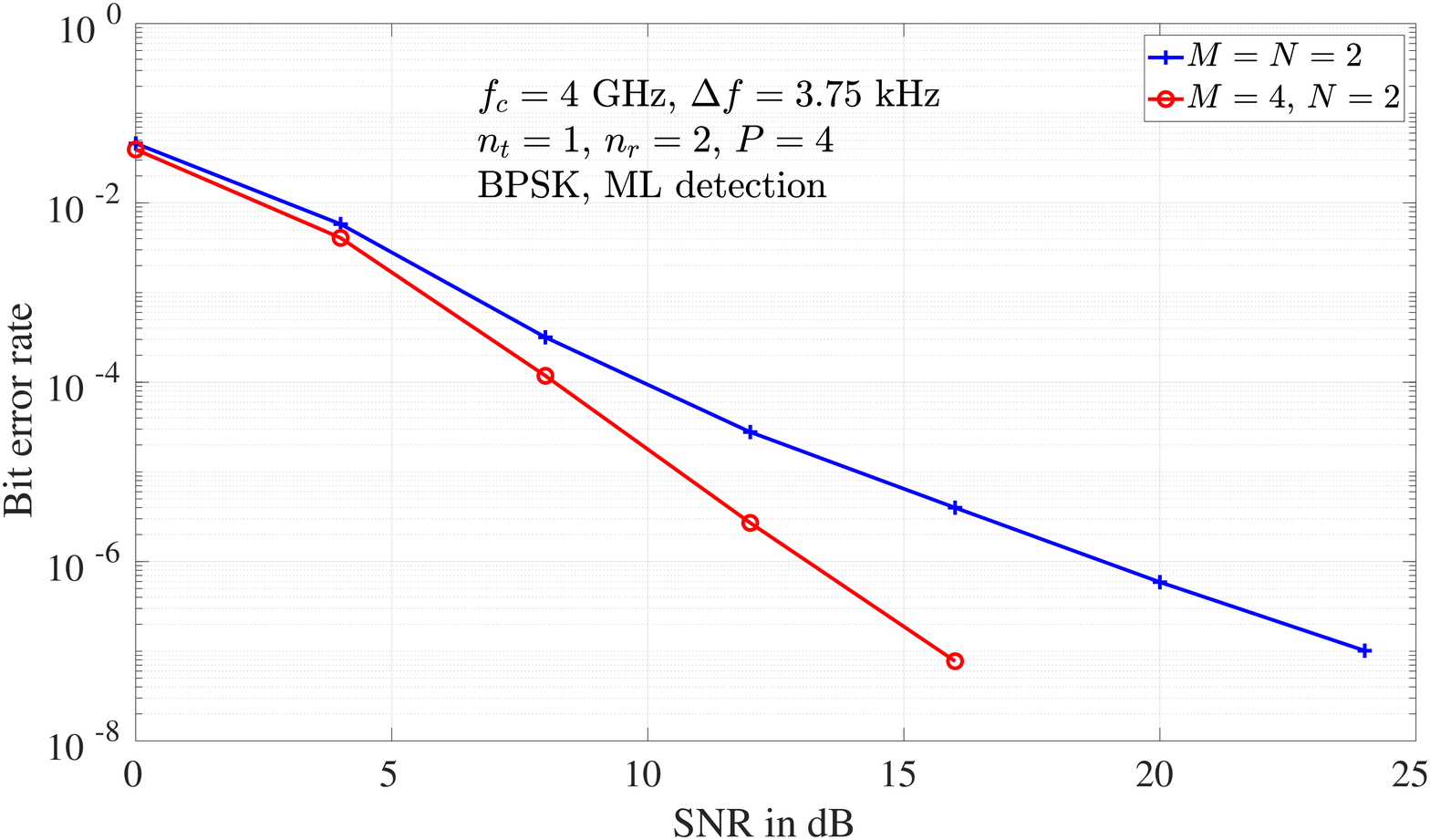}
\vspace{-4mm}
\caption{BER performance of $1\times 2$ OTFS system with $i)$ $M=N=2$ 
and $ii)$ $M=4$, $N=2$.}
\vspace{-5mm}
\label{MIMO_BER2}
\end{figure}

\begin{figure}
\centering
\includegraphics[width=9.5cm, height=6.5cm]{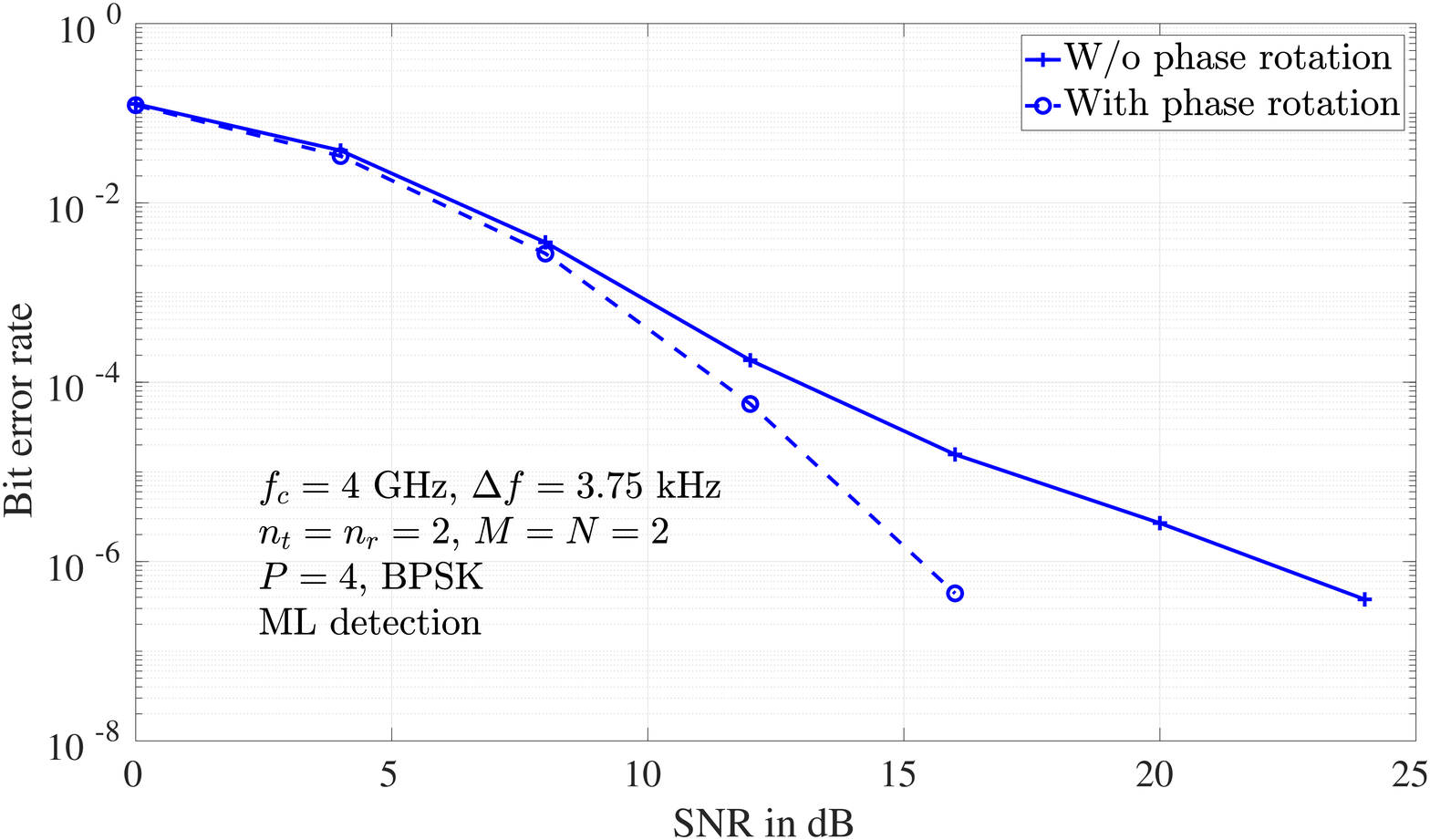}
\vspace{-4mm}
\caption{BER performance of $2\times 2$ MIMO-OTFS system without and with 
phase rotation, $M=N=2$.}
\label{PR_MIMO_OTFS}
\vspace{-6mm}
\end{figure}

\section{Conclusions}
\label{sec6}
We investigated the diversity of OTFS modulation and showed that the 
asymptotic diversity order of OTFS in a SISO setting with ML detection 
is one. Though the asymptotic diversity order is one, it was found that 
higher diversity performance can be achieved in the finite SNR regime before 
the diversity one regime takes over, and that the diversity one regime 
starts at lower BER values for increased frame sizes. These observations
were illustrated through a BER lower bound derived based on diversity one
PEPs and simulations. Next, a phase rotation scheme using transcendental 
numbers was proposed to extract the full diversity offered by the 
delay-Doppler channel. It was proved that the proposed phase rotation
achieves full diversity. Finally, we extended the diversity analysis
and results for MIMO-OTFS without and with phase rotation. 
Timing/frequency offset and synchronization effects and link adaptation 
in OTFS can be considered for future work. More robust systems targeting 
ultra-reliable and low-latency communication can be considered with proper 
configuration for $M$, $N$, coding, and the consideration of the natural 
presence of residual synchronization effects in combination with the 
proposed rotation scheme. 

\appendices
\section{Diversity Analysis for Non-zero Fractional Delays and Dopplers}
\label{app_a}
Recall the channel representation in the delay-Doppler domain denoted by 
$h(\tau,\nu)$ in (\ref{sparsechannel}). Consider the case of non-zero
fractional delays and Dopplers, i.e., consider
\begin{equation}
\tau_i=\frac{\alpha_i+a_i}{M\Delta f}, \ \ \mbox{and} \ \ \nu_i=
\frac{\beta_i+b_i}{NT},
\end{equation}
where $\alpha_i=[\tau_iM\Delta f]^{\odot}$, $\beta_i=[\nu_iNT]^\odot$, and 
$[.]^\odot$ denotes the nearest integer operator (i.e., rounding operator). 
Note that $\alpha_i$ and $\beta_i$ are integers corresponding to the indices 
of the delay $\tau_i$ and Doppler frequency tap $\nu_i$, respectively, and 
$a_i$, $b_i$ are the fractional delay and Doppler such that 
$-\frac{1}{2}< a_i,b_i,\leq \frac{1}{2}$. With this, we now proceed to 
derive the input-output relation for OTFS modulation taking into account 
the fractional part of the delay and Doppler shifts. Substituting 
(\ref{sparsechannel}) and (\ref{winfunc}) into (\ref{circ_conv}), and 
assuming rectangular window functions, we get
\begin{eqnarray}
h_w(\tau,\nu) & = &  \sum_{i=1}^{P}h_ie^{-j2\pi \tau_i \nu_i}w(\nu-\nu_i,
\tau-\tau_i) \nonumber \\
& = &  \sum_{i=1}^{P}h_ie^{-j2\pi \tau_i \nu_i}\mathcal{G}(\nu,\nu_i)
\mathcal{F}(\tau,\tau_i),
\label{Hw1}
\end{eqnarray}
where
\begin{eqnarray}
\mathcal{G}(\nu,\nu_i) & \triangleq & \sum_{n'=0}^{N-1} e^{-j2\pi (\nu-\nu_i)n'T}, 
\label{G} \nonumber \\
\mathcal{F}(\tau,\tau_i) & \triangleq & \sum_{m'=0}^{M-1} e^{j2\pi (\tau-\tau_i)m'
\Delta f}.
\label{F}
\end{eqnarray}
In order to use $h_w(\tau,\nu)$ of (\ref{Hw1}) in the OTFS input-output 
relation in (\ref{otfsinpoutp}), we need to evaluate $h_w(\tau,\nu)$ at 
$\nu=\frac{k-k'}{NT}, \tau=\frac{l-l'}{M\Delta f}$. Evaluating 
$\mathcal{G}(\nu,\nu_i)$ at $\nu=\frac{k-k'}{NT}$, we get
\begin{eqnarray}
\mathcal{G}\left(\frac{k-k'}{NT},\nu_i \right) & = & \sum_{n'=0}^{N-1} 
e^{-j\frac{2\pi}{N} (k-k'-\beta_i-b_i)n'}  \nonumber \\
& = & \frac{e^{-j2\pi(k-k'-\beta_i-b_i)}-1}{e^{-j\frac{2\pi}{N}(k-k'-
\beta_i-b_i)}-1}.
\label{G1}
\end{eqnarray} 
Note that due to the fractional Doppler $b_i$, for a given $k$,
$\mathcal{G}\left(\frac{k-k'}{NT},\nu_i \right)\neq 0$, for all $k'$.
It has been shown in \cite{otfs4} that the magnitude of $\mathcal{G}
\left(\frac{k-k'}{NT},\nu_i \right)$ has a peak at $k'=k-\beta_i$
and decreases as $k'$ moves away from $k-\beta_i$. Similarly, evaluating 
$\mathcal{F}(\tau,\tau_i)$ at $\tau=\frac{l-l'}{M\Delta f}$, we get
\begin{eqnarray}
\mathcal{F}\left(\frac{l-l'}{M\Delta f},\tau_i \right) & = & \sum_{m'=0}^{M-1} 
e^{j\frac{2\pi}{M} (l-l'-\alpha_i-a_i)m'} \nonumber \\
& = & \frac{e^{j2\pi(l-l'-\alpha_i-a_i)}-1}{e^{j\frac{2\pi}{M}(l-l'-
\alpha_i-a_i)}-1}.
\label{F1}
\end{eqnarray} 
Note that due to the fractional delay $a_i$, for a given $l$,
$\mathcal{F}\left(\frac{l-l'}{M\Delta f},\tau_i \right)\neq 0$, for all 
$l'$. Using the same argument used for 
$\mathcal{G}\left(\frac{k-k'}{NT},\nu_i \right)$, it follows that the 
magnitude of $\mathcal{F} \left(\frac{l-l'}{M\Delta f},\tau_i \right)$ 
has a peak at $l'=l-\alpha_i$ and decreases as $l'$ moves away from 
$l-\alpha_i$. Now, using (\ref{G1}), (\ref{F1}), and (\ref{Hw1}) in 
(\ref{otfsinpoutp}), we get
\begin{align}
y[k,l] = & \sum_{i=1}^{P}\sum_{q=0}^{M-1}\sum_{q'=0}^{N-1}
\left(\frac{e^{j2\pi(-q-a_i)}-1}{Me^{j\frac{2\pi}{M}(-q-a_i)}-M}\right) \nonumber \\ 
& \left(\frac{e^{-j2\pi(-q'-b_i)}-1}
{Ne^{-j\frac{2\pi}{N}(-q'-b_i)}-N}\right) h_ie^{-j2\pi \tau_i
\nu_i} \nonumber \\
&  x[(k-\beta_i+q')_N,(l-\alpha_i+q)_M].
\label{inpoutfracdeldop}
\end{align}
The input-output equation in (\ref{inpoutfracdeldop}) can be written
in vectorized form as
\begin{equation}
\mathbf{y}=\mathbf{Hx}+\mathbf{v},
\label{Hx_frac_del_dop}
\end{equation}
where $\mathbf{x}$, $\mathbf{y}$, $\mathbf{v}\in \mathbb{C}^{MN\times 
1}$, $\mathbf{H}\in \mathbb{C}^{MN\times MN}$, and the elements of 
$\mathbf{x}$, $\mathbf{y}$, and $\mathbf{H}$ are determined from
(\ref{inpoutfracdeldop}).

\subsection{Diversity analysis}
\label{app_b}
The vectorized input-output relation
in (\ref{Hx_frac_del_dop}) can be rewritten in an alternate form as
\begin{equation}
\mathbf{y}^T=\mathbf{h}'\mathbf{X}+\mathbf{v}^T,
\label{hXform1}
\end{equation}
where $\mathbf{y}^T$ is $1\times MN$ received vector, $\mathbf{h}'$
is a $1\times P$ vector whose $i$th entry is given by $h_ie^{-j2\pi 
\tau_i \nu_i}$, and $\mathbf{X}$ is a $P\times MN$ matrix whose 
$i$th column ($i=k+Nl$,\ $i=0,1,\cdots,MN-1$), denoted by $\mathbf{X}
[i]$, is given by (\ref{X_mat1}).
\begin{figure*}[t]
\begin{small}
\begin{equation}
\hspace{1mm}
\mathbf{X}[i]=\begin{bmatrix}
\displaystyle \sum_{q=0}^{M-1}\sum_{q'=0}^{N-1}
\left(\frac{e^{j2\pi(-q-a_1)}-1}{Me^{j\frac{2\pi}{M}(-q-a_1)}-M}\right)\left(\frac{e^{-j2\pi(-q'-b_1)}-1}
{Ne^{-j\frac{2\pi}{N}(-q'-b_1)}-N}\right)x[(k-\beta_1+q')_N,(l-\alpha_1+q)_M] \\
\displaystyle \sum_{q=0}^{M-1}\sum_{q'=0}^{N-1}
\left(\frac{e^{j2\pi(-q-a_2)}-1}{Me^{j\frac{2\pi}{M}(-q-a_2)}-M}\right)\left(\frac{e^{-j2\pi(-q'-b_2)}-1}
{Ne^{-j\frac{2\pi}{N}(-q'-b_2)}-N}\right)x[(k-\beta_2+q')_N,(l-\alpha_2+q)_M] \\
\vdots \\
\displaystyle \sum_{q=0}^{M-1}\sum_{q'=0}^{N-1}
\left(\frac{e^{j2\pi(-q-a_P)}-1}{Me^{j\frac{2\pi}{M}(-q-a_P)}-M}\right)\left(\frac{e^{-j2\pi(-q'-b_P)}-1}
{Ne^{-j\frac{2\pi}{N}(-q'-b_P)}-N}\right)x[(k-\beta_P+q')_N,(l-\alpha_P+q)_M]
\end{bmatrix}
\label{X_mat1}. 
\end{equation}
\end{small}
\end{figure*}

The representation of $\mathbf{X}$ in the form given in (\ref{hXform1}) 
allows us to view $\mathbf{X}$ as a $P\times MN$ symbol matrix. For 
convenience, we normalize the elements of $\mathbf{X}$ so that the 
average energy per symbol time is one. The SNR, denoted by $\gamma$, is 
therefore given by $\gamma=1/N_0$. Assuming perfect channel state 
information and ML detection at the receiver, the PEP between 
$\mathbf{X}_i$ and $\mathbf{X}_j$ is given by 
\begin{equation}
P(\mathbf{X}_i\rightarrow \mathbf{X}_j|\mathbf{h'},\mathbf{X}_i)=Q \left( \sqrt{\frac{\|\mathbf{h'}(\mathbf{X}_i-\mathbf{X}_j)\|^2}{2N_0}} \right).
\label{PEP1_1} 
\end{equation}
The PEP averaged over the channel statistics is given by
\begin{equation}
\small
P(\mathbf{X}_i\rightarrow \mathbf{X}_j)=\mathbb{E}
 \left[ Q \left( \sqrt{\frac{\gamma\ \|\mathbf{h'}
(\mathbf{X}_i-\mathbf{X}_j)\|^2}{2}} \right) \right].
\label{PEP2_1}
\end{equation}
As in Sec. \ref{sec3}, (\ref{PEP2_1}) can be obtained as 
\begin{equation}
P(\mathbf{X}_i\rightarrow \mathbf{X}_j) = \mathbb{E} \left[Q\left(\sqrt{\frac{\gamma \ \sum_{l=1}^r \lambda_l^2|\tilde{h}_l|^2}{2}}\right)\right],
\label{PEP3_1}
\end{equation}
where $r$ denotes the rank of the difference matrix 
$\mathbf{\Delta}_{ij}=(\mathbf{X}_i-\mathbf{X}_j)$, $\tilde{h}_l$ is the 
$l$th element of the vector $\tilde{\bf h}^H=\mathbf{U}^H\mathbf{h'}^H$, where
the matrix $(\mathbf{X}_i-\mathbf{X}_j) (\mathbf{X}_i-\mathbf{X}_j)^H$ is
Hermitian matrix that is diagonalizable by unitary transformation and hence
can be written as 
$(\mathbf{X}_i-\mathbf{X}_j)(\mathbf{X}_i-\mathbf{X}_j)^H=\mathbf{U\Lambda U}^H$,
where $\mathbf{U}$ is unitary and 
$\mathbf{\Lambda}=\mbox{diag}\lbrace \lambda_1^2, \cdots \lambda_P^2 \rbrace$,
$\lambda_i$ being $i$th singular value of the  difference matrix
$\mathbf{\Delta}_{ij}$,
The average PEP in (\ref{PEP3_1}) can be simplified to get the following 
upper bound on PEP \cite{DTse}
\begin{equation}
P(\mathbf{X}_i\rightarrow \mathbf{X}_j)  \leq \prod 
\limits_{l=1}^{r}\frac{1}{1+\ \dfrac{\gamma \lambda_l^2}{4P}}, 
\label{PEP4_1}
\end{equation}
which, at high SNRs, can be further simplified as
\begin{equation}
P(\mathbf{X}_i\rightarrow \mathbf{X}_j)  \leq \frac{1}{\gamma^r\prod 
\limits_{l=1}^{r} \dfrac{\lambda_l^2}{4P}}. 
\label{PEP5_1}
\end{equation}
From (\ref{PEP5_1}), it can be seen that the exponent of the SNR term 
$\gamma$ is $r$, which is equal to the rank of the difference matrix
$\mathbf{\Delta}_{ij}$. For all $i,j$, $i\neq j$, the PEP with the 
minimum value of $r$ dominates the overall BER. Therefore, the achieved 
diversity order, denoted by $\rho_{\tiny \mbox{siso-otfs}}$, is given by 
\begin{equation}
\rho_{\tiny \mbox{siso-otfs}} = \min_{i,j \ i\neq j}\ \mbox{rank}(\mathbf{\Delta}_{ij}).
\label{div_order_1}
\end{equation}
Now, consider a case when $x_i[k,l]=a$ and $x_j[k,l]=a'$, 
$\forall k=0,\cdots,N-1$ and $l=0,\cdots,M-1$. Then, 
$\mathbf{\Delta}_{ij}=(\mathbf{X}_i-\mathbf{X}_j)$ will be of the form 
$(a-a').\mathbf{Z}_{P \times MN}$, where each column of $\mathbf{Z}$ is 
identical and of the form given by
{\footnotesize
\begin{align}
(a-a')\begin{bmatrix}
\displaystyle \sum_{q=0}^{M-1}\sum_{q'=0}^{N-1}
\left(\tfrac{e^{j2\pi(-q-a_1)}-1}{Me^{j\frac{2\pi}{M}(-q-a_1)}-M}\right)\left(\tfrac{e^{-j2\pi(-q'-b_1)}-1}
{Ne^{-j\tfrac{2\pi}{N}(-q'-b_1)}-N}\right)\\
\displaystyle \sum_{q=0}^{M-1}\sum_{q'=0}^{N-1}
\left(\tfrac{e^{j2\pi(-q-a_2)}-1}{Me^{j\frac{2\pi}{M}(-q-a_2)}-M}\right)\left(\tfrac{e^{-j2\pi(-q'-b_2)}-1}
{Ne^{-j\frac{2\pi}{N}(-q'-b_2)}-N}\right) \\
\vdots \\
\displaystyle \sum_{q=0}^{M-1}\sum_{q'=0}^{N-1}
\left(\tfrac{e^{j2\pi(-q-a_P)}-1}{Me^{j\frac{2\pi}{M}(-q-a_P)}-M}\right)\left(\tfrac{e^{-j2\pi(-q'-b_P)}-1}
{Ne^{-j\tfrac{2\pi}{N}(-q'-b_P)}-N}\right)
\end{bmatrix}.
\label{Zmat}
\end{align}
}
Since all columns of $\mathbf{Z}$ are identical (independent of $k$ and $l$) 
with the form (\ref{Zmat}), rank of $\mathbf{Z}$ is clearly one, which is 
the minimum rank of $ \mathbf{\Delta}_{ij}$, $\forall i,j$, $i\neq j$. 
Hence, the asymptotic diversity order of OTFS with ML detection in the 
case of fractional delay and Dopplers is also one. 
\hspace{25mm} $\square$

{\em Simulation results:}
Figure \ref{SISO_BER_frac_delayDop} shows the BER performance with 
non-zero fractional delays and Dopplers. The figure shows the performance 
of two systems, $i)$ system-1 with $M=N=2$ and $ii)$ system-2 with $M=4$ 
and $N=2$. The carrier frequency and the subcarrier spacing used are 4 GHz 
and 3.75 kHz, respectively. Both the systems use BPSK and ML detection. 
A channel with $P=4$ paths with a maximum Doppler of 1.875 kHz (which 
corresponds to a speed of 506.25 km/h at 4 GHz carrier frequency), 
exponential power delay profile, and Jakes Doppler spectrum \cite{spectrum}
is considered.
The input-output relation in (\ref{inpoutfracdeldop}) which considers
the fractional part of the delay and Doppler values is used for the
simulations. From Fig. \ref{SISO_BER_frac_delayDop}, it is evident that 
the asymptotic diversity order of OTFS modulation is one in the case of 
non-zero fractional delays and Dopplers. Also, the asymptotic diversity 
order of one is achieved at lower BER values for increased values of $M$ 
and $N$. This behavior is the same as that observed in Sec. \ref{sec3}, 
where analysis and simulations were carried out without considering 
fractional delay and Doppler values. 

\begin{figure}[t]
\centering
\includegraphics[width=9.5cm, height=6.5cm]
{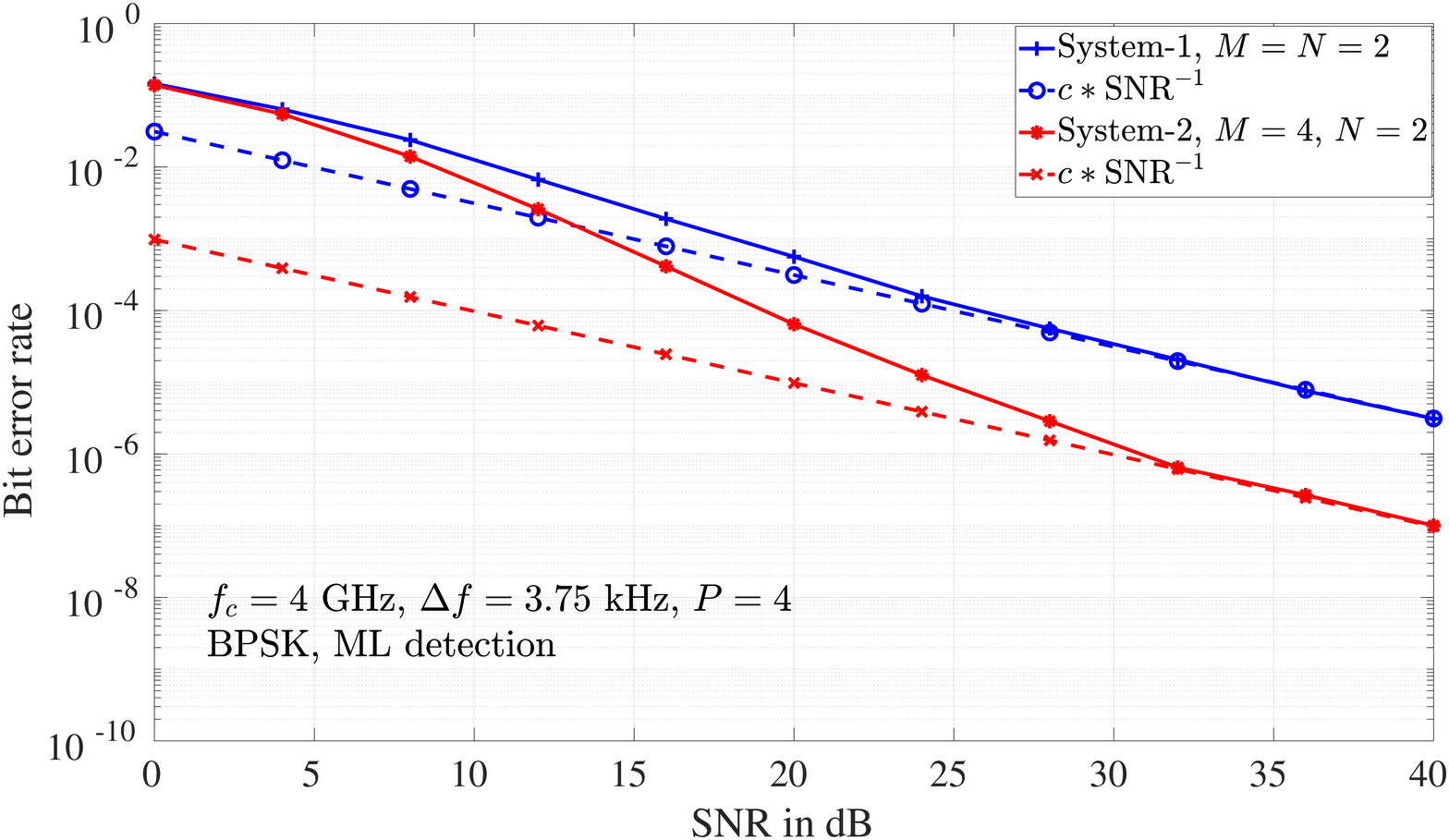}
\vspace{-2mm}
\caption{BER performance of OTFS for $i)$ $M=2$, $N=2$ and $ii)$ $M=4$, 
$N=2$, considering non-zero fractional delay and Doppler values.}
\vspace{-4mm}
\label{SISO_BER_frac_delayDop}
\end{figure}


\begin{thebibliography}{1}
\bibitem{jakes}
W. C. Jakes, {\em Microwave Mobile Communications}, New York: IEEE 
Press, reprinted, 1994. 

\bibitem{ofdm1}
T. Wang, J. G. Proakis, E. Masry, and J. R. Zeidler, ``Performance 
degradation of OFDM systems due to Doppler spreading,'' {\em IEEE 
Trans. Wireless Commun.}, vol. 5, no. 6, pp. 1422-1432, Jun. 2006.

\bibitem{otfs2}
R. Hadani, S. Rakib, M. Tsatsanis, A. Monk, A. J. Goldsmith, 
A. F. Molisch, and R. Calderbank, ``Orthogonal time frequency space 
modulation,'' in {\em Proc. IEEE WCNC'2017}, pp. 1-7, Mar. 2017.

\bibitem{otfs1a}
R. Hadani, S. Rakib, S. Kons, M. Tsatsanis, A. Monk, C. Ibars, J. Delfeld, 
Y. Hebron, A. J. Goldsmith, A. F. Molisch, and R. Calderbank, ``Orthogonal 
time frequency space modulation,'' arXiv:1808.00519v1 [cs.IT] 1 Aug 2018.

\bibitem{otfswhitepaper}
R. Hadani and A. Monk, ``OTFS: A new generation of modulation addressing 
the challenges of 5G,'' arXiv:1802.02623 [cs.IT] 7 Feb 2018.

\bibitem{otfs1}
A. Monk, R. Hadani, M. Tsatsanis, and S. Rakib, ``OTFS - orthogonal time
frequency space: a novel modulation technique meeting 5G high mobility 
and massive MIMO challenges,'' arXiv:1608.02993 [cs.IT] 9 Aug 2016.

\bibitem{otfs3}
R. Hadani, S. Rakib, A. F. Molisch, C. Ibars, A. Monk, M. Tsatsanis, 
J. Delfeld, A. Goldsmith, and R. Calderbank, ``Orthogonal time frequency 
space (OTFS) modulation for millimeter-wave communications systems,''
in {\em Proc. IEEE MTT-S Intl. Microwave Symp.}, pp. 681-683, Jun. 2017.

\bibitem{otfs6}
L. Li, H. Wei, Y. Huang, Y. Yao, W. Ling, G. Chen, P. Li, and Y. Cai,
``A simple two-stage equalizer with simplified orthogonal time frequency 
space modulation over rapidly time-varying channels,'' 
arXiv:1709.02505v1 [cs.IT] 8 Sep 2017.

\bibitem{otfs7}
A. Farhang, A. R. Reyhani, L. E. Doyle, and B. Farhang-Boroujeny,
``Low complexity modem structure for OFDM-based orthogonal time frequency 
space modulation,'' {\em IEEE Wireless Commun. Lett.}, vol. 7, no. 3, 
pp. 344-347,  Jun. 2018.

\bibitem{otfs4}
P. Raviteja, K. T. Phan, Y. Hong, and E. Viterbo, ``Interference
cancellation and iterative detection for orthogonal time frequency
space modulation,'' {\em IEEE Trans. Wireless Commun.}, vol. 17,
no. 10, pp. 6501-6515, Aug. 2018.

\bibitem{otfs5}
K. R. Murali and A. Chockalingam, ``On OTFS modulation for high-Doppler 
fading channels,'' in {\em Proc. ITA'2018}, Feb. 2018.

\bibitem{otfs9}
M. K. Ramachandran and A. Chockalingam, ``MIMO-OTFS in high-Doppler fading 
channels: signal detection and channel estimation,'' in {\em Proc. IEEE 
GLOBECOM'2018}, Dec. 2018. Online: arXiv:1805.02209v1 [cs.IT] 6 May 2018.

\bibitem{mtone1}
R. Nissel, S. Schwarz, and M. Rupp, ``Filter bank multicarrier modulation 
schemes for future mobile communications,'' {\em IEEE J. Sel. Areas Commun.}, 
vol. 35, no. 8, pp. 1768-1782, Aug. 2017.

\bibitem{gfdm1}
N. Michailow, M. Matthe, I. S. Gaspar, A. N. Caldevilla, L. L. Mendes, 
A. Festag, and G. Fettweis, ``Generalized frequency division multiplexing
for 5th generation cellular networks,'' {\em IEEE Trans. Commun.}, vol. 62, 
no. 9, pp. 3045-3061, Sep. 2014.

\bibitem{otfs_gfdm}
A. Nimr, M. Chafii, M. Matthe, and G. Fettweis, ``Extended GFDM framework: 
OTFS and GFDM comparison,'' online: arXiv:1808.01161v1 [eess.SP] 3 Aug 2018.

\bibitem{pulse_design1}
T. Strohmer and S. Beaver, ``Optimal OFDM design for time-frequency
dispersive channels,'' {\em IEEE Trans. Commun.}, vol. 51, no. 7, pp. 
1111-1122, Jul. 2003.

\bibitem{pulse_design2}
H. Bolcskei, P. Duhamel, and R. Hleiss, ``Design of pulse shaping 
OFDM/OQAM systems for high data-rate transmission over wireless 
channels,'' in {\em Proc. IEEE ICC'1999}, Jun. 1999.

\bibitem{Heisenberg}
S. D. Howard, A. R. Calderbank, and W. Moran, ``The finite Heisenberg-Weyl 
groups in radar and communications,'' {\em EURASIP Journal on 
Applied Signal Processing}, Jan. 2006.

\bibitem{rep_theory}
P. Jung, ``Weyl-Heisenberg representations in communication theory,''
Ph.D. dissertation, Technische Universitat, Berlin, 2007.

\bibitem{DSFT}
M. Dorfler and B. Torresani, ``Representation of operators in the 
time-frequency domain and generalized Gabor multipliers,'' online: 
arXiv:0809.2698v1 [math.AP] 16 Sep 2008.

\bibitem{DDchannelest}
A. Fish, S. Gurevich, R. Hadani, A. M. Sayeed, and O. Schwartz, 
``Delay-Doppler channel estimation in almost linear complexity,''
{\em IEEE Trans. Inform. Theory}, vol. 59, no. 11, pp. 7632-7644, 
Nov. 2013.

\bibitem{DTse}
D. Tse and P. Viswanath, {\em Fundamentals of Wireless Communication}, 
Cambridge University Press, 2005.

\bibitem{spectrum}
F. Hlawatsch and G. Mats, {\em Wireless Communications over Rapidly 
Time-Varying Channels}, Academic Press, 2011.

\bibitem{ieee802.11p}
A. M. S. Abdelgader and W. Lenan, ``The physical layer of the IEEE 802.11p 
WAVE communication standard: the specifications and challenges,''
{\em Proceedings of the World Congress on Engineering and Computer Science}, 
Oct. 2014.

\bibitem{iter_dec_otfs}
T. Zemen, M. Hofer, D. Löschenbrand, and C. Pacher, ``Iterative 
detection for orthogonal precoding in doubly selective channels,'' 
in {\em Proc. IEEE PIMRC'2018}, Sep. 2018.

\bibitem{davis}
P. J. Davis, {\em Circulant Matrices}, American Mathematical Society, 2012.

\bibitem{stbc_numthry}
M. O. Damen, A. Tewfik, and J. C. Belfiore, ``A construction of a space–time 
code based on number theory,'' {\em IEEE Trans. Inform. Theory}, vol. 48, 
no. 3, pp. 753-760, Mar. 2002.

\end{thebibliography}
\end{document}